\documentclass[reqno,a4paper,11pt]{amsart}

\usepackage{amsmath,amstext,amsfonts,amsbsy,eucal,amssymb}
\usepackage[latin1]{inputenc}

\numberwithin{equation}{section}

\newtheorem{theorem}{Theorem}[section]
\newtheorem{lemma}[theorem]{Lemma}

\newtheorem{proposition}[theorem]{Proposition}

\newtheorem{conjecture}[theorem]{Conjecture}

\theoremstyle{definition}

\newtheorem{example}[theorem]{Example}
\newtheorem{remark}[theorem]{Remark}

\begin{document}

\parskip 4pt
\baselineskip 16pt


\title[Volterra map and related recurrences]
{Volterra map and related recurrences}

\author[Andrei K. Svinin]{Andrei K. Svinin}
\address{Matrosov Institute for System Dynamics and Control Theory of 
Siberian Branch of Russian Academy of Sciences,
P.O. Box 292, 664033 Irkutsk, Russia}
\email{svinin@icc.ru}

\date{\today}



\begin{abstract}
In this paper we analyze recent work \cite{Hone1} by Hone,  Roberts and  Vanhaecke, where the so-called  Volterra map was introduced via the Lax equation that looks  similar to the Lax representation for the Mumford's system \cite{Vanhaecke}.   This map turns out to be birational and a corresponding dynamical system on an affine space $M_g$ of dimension $3g+1$ was associated with it. This mapping  is related to some discrete equation of  the order $2g+1$ associated with the Stieltjes continued fraction expansion of a certain function on a hyperelliptic (elliptic) curve of genus $g\geq 1$.   The authors of the paper provides examples of this equation for the simplest cases $g=1$ and $g=2$, but for higher values of $g$, corresponding equation turns out to be too cumbersome to write them out.

We present an approach in which the mentioned $(2g+1)$-order equation can be written out for all values of $g\geq 1$ in a compact form. This equation is not new and can be found, for example, in \cite{Svinin3}. An essential point in our framework is the use of special class of discrete polynomials which as shown to be closely related to the Stieltjes continued fraction. On the one hand, this allows us to generalize some of the results of the work \cite{Hone1}.  On the other hand, many things in this approach can be presented in a more compact and unified form. Ultimately, we believe that this allows us to give a new perspective on this topic.
\end{abstract}

\maketitle

\section{Introduction}

This work was inspired by a recent work \cite{Hone1}, where the authors introduced the so-called  Volterra map. More precisely, it  is understood as an infinite family of mappings parametrized by an integer $g\geq 1$. It is most conveniently defined using the Lax equation,   which is a certain matrix relation $L M=M \tilde{L}$,  containing a triple of polynomials $\left(P(x), Q(x), R(x)\right)$ in the formal variable $x$. It should be noted that a Lax equation with some parameter  is common in the theory of integrable equations since it gives some conserved quantities for the corresponding continuous or discrete map.  The authors of the paper showed that the Volterra map is birational and therefore can be iterated in both directions. As a result of an infinite number of iterations, it  arises a  discrete dynamical system  on the affine space $M_g$ of dimension $3g+1$. It should be noted that definition of the Volterra map using the traceless matrix $L(x)$ is very reminiscent of Mumford's design \cite{Mumford}, \cite{Vanhaecke}. 

The original motivation for the work \cite{Hone1}, as the authors note, was to study a recurrence
\begin{eqnarray}
&&u_nu_{n+1}u_{n+2}+u_{n+2}u_{n+3}u_{n+4}+2u_{n+2}^2\left(u_{n+1}+u_{n+3}\right) \nonumber \\
&&\;\;\;\;\;+u_{n+2}\left(u_{n+1}^2+u_{n+1}u_{n+3}+u_{n+3}^2\right)+u_{n+2}^3 - H_1 u_{n+2}\left(u_{n+1}+u_{n+2}+u_{n+3}\right) \nonumber \\
&&\;\;\;\;\;+H_2 u_{n+2}-\mathcal{G}_{2, 0}=0
\label{675230976}
\end{eqnarray} 
appeared as a result of some classification of the integrable maps in   \cite{Gubbiotti}. This is one of the equations associated with the four-dimensional birational map $\mathbb{C}^4 \mapsto \mathbb{C}^4$ with two independent invariants
\[
\mathcal{G}_{2, 1}= -u_{1}u_2\left(u_0u_1+u_2u_3+(u_1+u_2)^2-u_0u_3-H_1 (u_1+u_2)+H_2\right)+\mathcal{G}_{2, 0} \left(u_1+u_2\right)
\]  
and
\begin{eqnarray*}
\mathcal{G}_{2, 2}&=& -u_{1}u_2\left((u_0+u_1+u_2+u_3)(u_0u_1+u_2u_3+(u_1+u_2)^2) \right. \\
&& -H_1(u_0+u_1+u_2)(u_1+u_2+u_3) \\
&&\left. +H_2 (u_0+u_1+u_2+u_3)\right)+\mathcal{G}_{2, 0}\left(u_0u_1+u_2u_3+(u_1+u_2)^2\right).
\end{eqnarray*}
For brevity of writing, we sometimes use the notation $u_0, u_1,\ldots$ instead of $u_n, u_{n+1},\ldots$
\begin{remark}
The authors of the work \cite{Hone1} have noted that  recurrence (\ref{675230976}) is generated by the mapping denoted by $\left(\mathrm{P.iv}\right)$  in  \cite{Gubbiotti}. 
It represents one of the examples of four-dimensional birational mappings with prescribed properties.   
It is convenient for us to use our own notations dictated by the logic of the work from the very beginning. 
For comparison with  \cite{Hone1}, we note that  $H_1=-\nu$,  $H_2=b$  and $\mathcal{G}_{2, 0}=-a$.
Compared to the paper \cite{Hone1} we have $\mathcal{G}_{2, 1}=-K_1$ and $\mathcal{G}_{2, 2}=-K_2$. Note that the invariants $K_1$ and $K_2$ are given by  formulas (1.3) and (1.4) in  \cite{Hone1}, respectively. 
\end{remark}

Looking ahead a little, we  notice right away that the equation (\ref{675230976}) can be rewritten in a compact form
\[
S^3_3(n)-H_1 S^2_2(n+1) + H_2 S^1_1(n+2) -\mathcal{G}_{2, 0}=0
\]
or better yet, like 
\begin{equation}
\mathcal{G}_{2, 0}=u_{n+2}\left(S^2_4(n)-H_1 S^1_3(n+1) + H_2\right),
\label{983210987}
\end{equation}
where $S^k_s(n)$, for different values of $k, s\geq 1$, denote some  homogeneous polynomial in $(u_n, u_{n+1},\ldots)$.  These polynomials  will be explicitly defined in the Section \ref{873289000}. 
It looks highly likely that the fact that equation (\ref{675230976}) can be written as (\ref{983210987}) is not accidental and that there is something hidden behind this. In fact  the recurrence (\ref{675230976}) is essentially  known equation closely related to the differential-difference equation
\begin{equation}
\frac{d u_n}{d t}=u_n\left(u_{n+1}-u_{n-1}\right)
\label{983999837}
\end{equation}
known as the Volterra lattice  \cite{Kac}, \cite{Manakov}. It is also well known that  (\ref{983999837}) represents one of the  discretizations of the Korteweg-de Vries equation.

A characteristic property of the equation (\ref{983999837}) is the presence of a hierarchy  of pair-wise commuting generalized symmetries  that can be written, together with  the original equation (\ref{983999837}), in explicit form \cite{Svinin1}
\begin{equation}
\frac{d u_n}{d t_s}=u_n\left(S^s_s(n-s+2)-S^s_s(n-s)\right),\;\; s\geq 1,
\label{977657}
\end{equation}
as one sees, again using mentioned discrete polynomials $S^k_s(n)$. For clarity, let us write down the first few members of a sequence $(S^s_s(n))$ that define a hierarchy of the flows (\ref{977657}). They are 
\[
S^1_1(n)=u_n,\;\; S^2_2(n)=u_{n+1}\left(u_n+u_{n+1}\right)+u_{n+2}u_{n+1},
\]
\begin{eqnarray*}
S^3_3(n)&=&u_{n+2}\left(u_{n+1}(u_n+u_{n+1}+u_{n+2})+u_{n+2}(u_{n+1}+u_{n+2})+u_{n+3}u_{n+2}\right) \\
&&+u_{n+3}\left(u_{n+2}(u_{n+1}+u_{n+2})+u_{n+3}u_{n+2}\right)\\
&&+u_{n+4}u_{n+3}u_{n+2},\ldots
\end{eqnarray*}
More precisely, one can say that the relation (\ref{675230976}) is a constraint compatible  with any  flow determined by the evolutionary equation (\ref{977657}).

It turns out, as shown in  \cite{Hone1}, that the equation (\ref{675230976}) is birationally equivalent to some system of discrete equations of the dynamical system generated by the Volterra map $\mathcal{V}_g$, for the case $g=2$. 
The main goal of our work is to show such an  equation which is, on the one hand, birationally equivalent to the equations of the dynamical system associated with the Volterra map for any $g\geq 1$, and on the other hand is compatible with the flows of the integrable hierarchy of the Volterra lattice (\ref{977657}). 

The paper is organized as follows.  In the next section, we define, in the same way as in  \cite{Hone1}, the  Volterra map and corresponding dynamical system. We also discuss in this section how the Stieltjes continued fraction is related to the Volterra map. Next, in  Section \ref{873289}, we do some preliminary calculations in order to show the usefulness of  a certain class of homogeneous discrete polynomials $S^k_s(n)$. Namely, taking the simplest examples $g=1$ and  $g=2$, we look at how the coordinates of the phase space $M_g$ are expressed through the terms of a sequence $(u_n)$ satisfying some corresponding  recurrence. In  Section \ref{873289000}, we  define these polynomials and write out all the identities that which are needed in the future. In fact, the purpose of this section is more serious. Namely, here we propose some approach that allows us to represent in an explicit compact form the discrete equation of order $2g+1$ associated with the Volterra map $\mathcal{V}_g$. In section \ref{6752109654} we immediately write out this equation and analyze its invariants.  It turns out that this equation, or more precisely class of equations, is not new and  has been appeared in \cite{Svinin3} as a constraint consistent with the Volterra lattice (\ref{983999837}).  Section \ref{99832098} is devoted to homogeneous equations connected by some substitution with the mentioned equation of the order $2g+1$. For example, as shown in work \cite{Hone1}, in the case  $g=1$, this is the well-known Somos-5 recurrence. The most characteristic property of this equation, related to the integerness of some sequences, is the Laurent one. We show a homogeneous equation which is equivalent, by construction, to Somos-5, which also has the Laurent property, but which has its own peculiarities. As an example, we also write out homogeneous equations for cases $g=2$ and $g=3$. At the end we write out some examples of integer sequences generated by these homogeneous recurrences.

\section{Volterra map}

\subsection{Definition}

Let us consider a triple of polynomials
\begin{equation}
P(x)=1+\sum_{j=1}^g p_j x^j,\;\; Q(x)=2+\sum_{j=1}^g q_j x^j,\;\; R(x)=\sum_{j=1}^{g+1} r_j x^j
\label{87324}
\end{equation}
with coefficients belonging to some  field, say $\mathbb{C}$, and corresponding traceless Lax matrix
\[
L(x)=
\left(
\begin{array}{cc}
P(x) & R(x) \\
Q(x) & -P(x) 
\end{array}
\right).
\]
Let $M_g\cong\mathbb{C}^{3g+1}$ to be an affine space of dimension $3g+1$ whose points are determined by the coordinates $\xi=\left(p_{1},\ldots, p_{g}, q_{1},\ldots, q_{g},  r_{1},\ldots, r_{g+1}\right)$.  The first thing we would like to do is to define, as in  \cite{Hone1}, some mapping $\mathcal{V}_g : M_g\mapsto M_g$, which  turns out to be  birational. As in the paper \cite{Hone1}, we  call it the Volterra map. Perhaps the most convenient way to do this is the Lax equation representation. Let us define a linear combination
\begin{equation}
u=p_{1}-\frac{q_{1}}{2}+\frac{r_{1}}{2}
\label{897650003}
\end{equation}
and  write down the  matrix relation 
\begin{equation}
L(x)M(x)=M(x)\tilde{L}(x)
\label{8999876503}
\end{equation}
with 
\[
M(x)=
\left(
\begin{array}{cc}
1 & u x \\
1 & 0 
\end{array}
\right).
\]
As one can  see, (\ref{8999876503}) implies three relations
\begin{equation}
\tilde{P}(x)=Q(x)-P(x),\;\; \tilde{Q}(x)=\frac{2 P(x)-Q(x)+R(x)}{u x},\;\;  \tilde{R}(x)=u x Q(x).
\label{89769803}
\end{equation}
As in \cite{Hone1}, we call the mapping $(P, Q, R)\mapsto (\tilde{P}, \tilde{Q}, \tilde{R})$  the Volterra map and denote it by $\mathcal{V}_g$. Of course, we  immediately pay attention to the fact that if it turns out that $u$ given by (\ref{897650003}) is zero, then  the Volterra map is not defined.  Let us assume that for given triple $(P, Q, R)$ one has  $u\neq 0$. Then it is easy to verify that in this case three relations in (\ref{89769803}) are correctly defined and therefore indeed define some mapping.

Evidently, an inverse mapping to (\ref{89769803}) is represented by the formulas
\begin{equation}
P(x)=\frac{\tilde{R}(x)}{u x}-\tilde{P}(x),\;\; Q(x)=\frac{\tilde{R}(x)}{u x},\;\; R(x)=u x \tilde{Q}(x)-\frac{\tilde{R}(x)}{u x}+2 \tilde{P}(x).
\label{890003}
\end{equation}
Since $u=\tilde{r}_1/2$, we see that (\ref{890003}) is well defined and we therefore conclude that the Volterra mapping   is indeed birational. 
Let us emphasize once again that the Volterra map is defined only if $u\neq 0$. In this case, the triple $(P, Q, R)$ is said to be regular. Well, let this triple be regular, but we cannot be sure that triple $(\tilde{P}, \tilde{Q}, \tilde{R})$ is also regular. To our knowledge, there is no known condition for given triple $(P, Q, R)$ that would guarantee the regularity of all triples obtained by repeated  iterations of the Volterra map.

\subsection{Invariants}

Given a triple of polynomials $(P, Q, R)$, let us define a polynomial
\begin{equation}
f(x)=-\det L(x) = P^2(x)+Q(x)R(x)=1+\sum_{j=1}^{2g+1}c_j x^j.
\label{8976753203}
\end{equation}
We  use the same notation $c_k$ for the invariant and its numerical value in the hope that this will not lead to misunderstanding. 

From the Lax equation (\ref{8999876503}), it immediately follows that the finite set of coefficients $\left(c_1,\ldots, c_{2g+1}\right)$ represents  invariants for the Volterra map.   Taking into account (\ref{8976753203}), we may  write
\begin{equation}
c_k=\sum_{j=0}^k p_{j}p_{k-j}+\sum_{j=0}^{k-1} q_{j}r_{k-j}
\label{8973}
\end{equation}
with understanding that $p_{0}=1,\; q_{0}=2$ and  $p_{j}=q_{j}=r_{j+1}=0$ for all $j\geq g+1$.

\subsection{Stieltjes continued fraction}

Let us consider algebraic function $y=y(x)$ defined by the relation
\begin{equation}
y^2=f(x),
\label{89000000983}
\end{equation}
where the polynomial $f(x)$ is given by (\ref{8976753203}). We choose the sign of the solution of (\ref{89000000983}) so that the  series expansion for the function $y(x)$ has the form 
$y=1+\sum_{j\geq 1} H_j x^j$, where the coefficients $H_j$ are defined to be some polynomials in coordinates  of the point of the phase-space $M_g$. It should be  noted that, due to  (\ref{89000000983}), the  $H_k$'s  are also can be considered as polynomials in  the coefficients $c_k$. 
For example, 
\begin{equation}
H_1=\frac{c_1}{2},\;\; H_2=\frac{c_2}{2}-\frac{c_1^2}{8},\;\; H_3=\frac{c_3}{2}-\frac{c_1 c_2}{4}+\frac{c_1^3}{16},\ldots
\label{00985463}
\end{equation}

The proof of the following lemma can be found in \cite{Hone1}.
\begin{lemma}
The Volterra map $\mathcal{V}_g$ given by the Lax equation (\ref{8999876503}) implies the relation 
\begin{equation}
\left(y+P\right)\left(y+\tilde{P}\right)=Q \left(y+\tilde{P}\right)+u x Q \tilde{Q},
\label{890976543}
\end{equation}
where $u$ is given by (\ref{897650003}).
\end{lemma}
Let us rewrite  (\ref{890976543}), although in a more cumbersome, but on the other hand more meaningful form
\begin{equation}
\displaystyle
\frac{y+P}{Q}=1+\frac{u x}{\displaystyle\frac{y+\tilde{P}}{\tilde{Q}}}.
\label{8909999000965863}
\end{equation}
It is  clear that we can infinitely continue the fraction on the right side of the relation (\ref{8909999000965863}) to get
\begin{equation}
\frac{y+P}{Q}=1+\frac{u_0 x}{1+\displaystyle{\frac{u_1 x}{1+\displaystyle{\frac{u_2 x}{1+\cdots}}}}},
\label{89863}
\end{equation}
where $u_0=u,\; u_1=\tilde{u}, u_2=\tilde{\tilde{u}},\ldots$ As a result we get the Stieltjes continued fraction  representation \cite{Stieltjes} for the  series expansion of the rational function $\left(y+P\right)/Q$ defined on the curve $y^2=f(x)$ that follows from the characteristic equation $\det\left(L(x)-y\right)=0$. For example, we have 
\[
u_1=\frac{r_1+q_1}{2}-\frac{2p_2+r_2-q_2}{2p_1+r_1-q_1}.
\]
If $g=1$, we must put  $p_2=q_2=0$. Of course expressions for $u_2, u_3$ and so on are very cumbersome and there is no point in writing them out. It is only important to note that, for all $n\geq 0$, there exists unambiguous expression $u_n=R_n(\xi)$, with $R_n$ being some rational function of coordinates of a point in phase space $M_g$.

\subsection{Dynamical system}

Since the Volterra map is birational we are able to iterate it in both directions. By doing so, we get a bi-infinite sequence of points of the phase space $M_g$ with coordinates $\xi_n=\left(p_{n, 1},\ldots, p_{n, g}, q_{n, 1},\ldots, q_{n, g},  r_{n, 1},\ldots, r_{n, g+1}\right)$  uniquely determined by the Lax equation
\begin{equation}
L_n(x)M_n(x)=M_n(x)L_{n+1}(x)
\label{897653}
\end{equation}
with
\[
L_n(x)=
\left(
\begin{array}{cc}
P_n(x) & R_n(x) \\
Q_n(x) & -P_n(x) 
\end{array}
\right)\;\; \mbox{and}\;\;
M_n(x)= 
\left(
\begin{array}{cc}
1 & u_n x \\
1 & 0 
\end{array}
\right),
\]
where  
\begin{equation}
u_{n}=p_{n, 1}-\frac{q_{n, 1}}{2}+\frac{r_{n, 1}}{2}, 
\label{8899}
\end{equation} 
for all $n\in\mathbb{Z}$. Describing (\ref{897653}) in more detail, we get the relations
\[
P_{n+1}(x)=Q_{n}(x)-P_{n}(x),\;\; Q_{n+1}(x)=\frac{2 P_{n}(x)-Q_{n}(x)+R_{n}(x)}{u_{n} x},\;\; 
\]
\begin{equation} 
R_{n+1}(x)=u_{n} x Q_{n}(x).
\label{8877545}
\end{equation} 
Let us also rewrite, for convenience, the relations (\ref{8877545})  in terms of coordinates. We have
\begin{equation}
p_{n+1, k}=q_{n, k}-p_{n, k},\;\; q_{n+1, k}=\frac{2 p_{n, k+1}-q_{n, k+1}+r_{n, k+1}}{u_n},\;\; r_{n+1, k}= u_{n}q_{n, k-1}.
\label{880099875}
\end{equation}

In conclusion of this section let us  present the Theorem 3.8 from the work \cite{Hone1}. Let
\[
\frac{y+P(x)}{Q(x)}=1-\mathcal{H}(x),\;\; \mbox{where}\;\; \mathcal{H}(x)=\sum_{j\geq 1}(-1)^j h_j x^j.
\]
Substituting $y=-P+\left(1-\mathcal{H}\right) Q$ into the relation $y^2=P^2+QR$, we derive the generating relation
\[
2P-Q+R+2\left(Q-P\right)\mathcal{H}-Q \mathcal{H}^2=0.
\]
An analysis of this relationship leads to the  theorem that can be found in  \cite{Hone1}.  However, we formulate it a bit in our own way
\begin{theorem}
\label{897643098}
The generating relation (\ref{8973209}) implies an infinite number of equations
\[
h_n=\sum_{j=1}^g (-1)^j \left(p_j-q_j\right) h_{n-j}+\frac{1}{2}\sum_{k=0}^g (-1)^k q_k \sum_{j=1}^{n-k-1}h_jh_{n-j-k},\;\; n\geq g+2,
\]
where the set of the initial values $\left(h_1, h_2,\ldots, h_{g+1}\right)$  is related by a birational mapping $\left(h_1, h_2,\ldots, h_{g+1}\right)\leftrightarrow \left(r_1, r_2,\ldots, r_{g+1}\right)$ depending on the coefficients $p_k$ and $q_k$.
\end{theorem}

\section{Preliminary calculations}
\label{873289}

Theoretically, the following  is embedded in the formulas (\ref{890976543}) and (\ref{89863}). Let us define  an initial point $\xi$ of the phase space $M_g$ by specifying a triple $(P, Q, R)$. The formula (\ref{89863}) provides us all iterates  $u_n$ for $n\geq 1$. We would like to discuss in detail the nontrivial fact that the system of equations (\ref{880099875})  is birationally related to  some polynomial recurrence  
\begin{equation}
P(u_n,\ldots, u_{n+2g+1}; c_1,\ldots, c_g)=0,
\label{98430}
\end{equation}
where $u_n$ is defined by (\ref{8899}). 

Before we tackle the general case, let us first look at the simplest examples. In the case $g=1$ the system of equations (\ref{880099875}) become
\begin{equation}
p_{n+1, 1}=q_{n, 1}-p_{n, 1},\;\; q_{n+1, 1}=\frac{r_{n, 2}}{u_n},\;\; r_{n+1, 1}= 2 u_{n},\;\; r_{n+1, 2}= u_{n}q_{n, 1}.
\label{88000099099875}
\end{equation}
From the second and fourth equations in (\ref{88000099099875}),  it follows the relation $u_{n+1} q_{n+1, 1}=u_{n+2} q_{n+3, 1}$ which, in turn, using the first relation, can be rewritten as
\begin{equation}
u_{n+1} \left(p_{n+1, 1}+p_{n+2, 1}\right)=u_{n+2} \left(p_{n+3, 1}+p_{n+4, 1}\right).
\label{899985}
\end{equation}
Since $c_1=2p_{n, 1}+2r_{n, 1}$ and $r_{n, 1}=2 u_{n-1}$ then $p_{n, 1}=-2 u_{n-1}+c_1/2$. Substituting the latter into (\ref{899985}),  as a result, we get the recursion  
\begin{equation}
u_{n+1}\left(u_n+u_{n+1}-\frac{c_1}{2}\right)=u_{n+2}\left(u_{n+2}+u_{n+3}-\frac{c_1}{2}\right),
\label{87786328}
\end{equation} 
which is a special case of the relation (\ref{98430}).

Also in the work \cite{Hone1}, for the case $g=2$, a fifth-order equation was written out, which looks  more cumbersome than  (\ref{87786328}). It was also shown that for any  $g\geq 1$ there exists an equation of order $2g+1$ which is satisfied by the sequence $(u_n)$. As $g$ grows, this equation  looks even more cumbersome if we try to write it explicitly. One of the results of our work which we  formulate further in the form of a theorem is that this equation can be  briefly written in terms of some discrete polynomials the definition of which we will give in the next section. 

We will consider relations (\ref{880099875})  together with constraints
\[
c_k=\sum_{j=0}^k p_{n, j}p_{n, k-j}+\sum_{j=0}^{k-1} q_{n, j}r_{n, k-j}
\]
assuming that $g$ is large enough so as not to worry about having to zero something out.  To illustrate the general approach, we are going to do some preliminary calculations, the results of which  lead us to think that some special discrete polynomials constantly appear.  
More precisely we would like to show how $(p_{n, k}, q_{n, k}, r_{n, k})$ is expressed through the variables  $u_n$ and $c_k$.

Since $c_1=2p_{n, 1}+2r_{n, 1}$ and $r_{n, 1}=2 u_{n-1}$ then 
\[
p_{n, 1}=-2 u_{n-1}+\frac{c_1}{2}=-2 S^1_1(n-1)+H_1
\]
and
\[
q_{n, 1}=p_{n, 1}+p_{n+1, 1}=-2 \left(u_{n-1}+u_{n}\right)+c_1=-2 \left(S^1_2(n-1)-H_1\right),
\]
where  $H_1=c_1/2$. We use the notations like  $S^1_1(n)=u_n$ and $S^1_2(n)=u_n+u_{n+1}$ which seems redundant for now, however, later we will show what meaning these designations have.
Let us now consider the relation
\begin{equation}
c_2=2p_{n, 2}+p_{n, 1}^2+2r_{n, 2}+q_{n, 1}r_{n, 1},
\label{876098}
\end{equation}
where
\[
r_{n, 2}=u_{n-1}q_{n-1, 1}=-2 u_{n-1}\left(S^1_2(n-2)-H_1\right).
\]
Solving (\ref{876098}) in favor of $p_{n, 2}$, we get
\begin{equation}
p_{n, 2}=2 S_2^2(n-2)-2 H_1 S_1^1(n-1)+H_2,
\label{786500006}
\end{equation}
where we have denoted
\begin{eqnarray}
S^2_2(n)&=&u_{n+1}\left(u_n+u_{n+1}\right)+u_{n+2}u_{n+1} \nonumber \\
&=&u_{n+1}S^1_2(n)+u_{n+2}S^1_1(n+1). \label{90843298}
\end{eqnarray}
In addition we calculate
\begin{eqnarray}
q_{n, 2}&=&p_{n, 2}+p_{n+1, 2} \nonumber \\
&=&2\left(S^2_2(n-2)+S^2_2(n-1)\right) -2H_1 \left(S^1_1(n-1)+S^1_1(n)\right)+2H_2 \nonumber \\
&=&2\left(S^2_3(n-2) -H_1 S^1_2(n-1)+H_2\right),
\label{89732165}
\end{eqnarray}
where  $H_2=c_2/2-c_1^2/8$ and
\begin{equation}
S^2_3(n)=u_{n+1}S^1_3(n)+u_{n+2}S^1_2(n+1)+u_{n+3}S^1_1(n+2).
\label{789986}
\end{equation}
Note also that easily verifiable identities 
\begin{equation}
S^1_1(n)+S^1_1(n+1)=S^1_2(n)\;\; \mbox{and}\;\; S^2_2(n)+S^2_2(n+1)=S^2_3(n)
\label{7865887646}
\end{equation}
are used to get (\ref{89732165}). Of course, we can continue to make more and more cumbersome calculations, but it is better to make some assumptions and then to prove them.

\section{Homogeneous discrete polynomials}
\label{873289000}

After the examples given in the previous section, one feels that it is time to introduce an infinite class of homogeneous discrete polynomials $S^k_s(n)$ which, as can be seen, persistently appear in the context of the problem under discussion. These discrete polynomials    were introduced in \cite{Svinin1} and then  exploited in  \cite{Svinin2}, \cite{Svinin3}, \cite{Svinin4}. The following material can be found partially in these papers, and is provided below mainly for the sake of completeness. That is, in this section we would like to recall some of the results of the mentioned works, in order to  link them with the main content of this paper. In essence, we need this section in order to collect together the identities  for some discrete polynomials. It should be noted that some identities presented in this section are new. 

\subsection{Definitions}

For better understanding, it is better to assume in this section  that $(u_n)$ is a free bi-infinite sequence, that is, it need not satisfy any equation. Also, let us not assume in this section that $u_n$ is defined by the formula (\ref{8899}).

It may already be clear that $k$,  is the degree of the homogeneous polynomial $S^k_s(n)$. Let us assume, for convenience, that $S^0_s(n)=1$. The polynomial $S^1_s(n)$, for any $s\geq 1$, is defined as
\begin{equation}
S^1_s(n)=\sum_{j=0}^{s-1} u_{n+j}.
\label{188885}
\end{equation}
In turn the formulas of type  (\ref{90843298}) and (\ref{789986}) suggests the idea to define discrete polynomials $S^2_s(n)$, for any $s\geq 1$, by
\begin{equation}
S^2_s(n)=\sum_{j=0}^{s-1}u_{n+j+1}S^1_{s-j}(n+j).
\label{1600008699085}
\end{equation}
Then, inspiring by (\ref{1600008699085}), we  write down the general formula for $S^k_s(n)$ in the form of  a recurrence 
\begin{equation}
S^k_s(n)=\sum_{j=0}^{s-1} u_{n+j+k-1} S^{k-1}_{s-j}(n+j).
\label{6000098652}
\end{equation}
This recurrence  is good from a practical point of view because it allows us to write out  successively polynomials $S^k_s(n)$ with increasingly large values of  $k$, starting, for example,  from $S^1_s(n)$ given by (\ref{188885}), however, let us note that one can also write these polynomials in the following explicit form:
\begin{equation}
S^k_s(n)=\sum_{0\leq \lambda_1\leq \cdots\leq  \lambda_k\leq s-1} \prod_{j=1}^k u_{n+\lambda_j+k-j}. 
\label{67500000000}
\end{equation}
This formula  tells us  that $S^k_s(n)$ is a homogeneous polynomial in variables $\left(u_n,\ldots, u_{n+k+s-2}\right)$. 

So we have two definitions of discrete polynomials $S^k_s(n)$, but another one will also be useful to us. Namely, let us consider a pseudo-difference operator 
\begin{equation}
\mathcal{T}=\Lambda^{-1} + u_{n-1} x \Lambda^{-2},
\label{67888}
\end{equation}
where $\Lambda$ stands for the shift operator acting as  $\Lambda(f_n)=f_{n+1}$. Then we  consider its negative powers
\begin{equation}
\mathcal{T}^{-s}=\Lambda^s+\sum_{j\geq 1} (-1)^j x^j S^j_s(n-j+1) \Lambda^{s-j},\;\; s\geq 1.
\label{678760}
\end{equation}
This expression is understood as a formal pseudo-differential operator or, if you like, a symbol. Note that in expression (\ref{678760}) there is a finite number of terms containing positive powers of  $\Lambda$ and an infinite number of terms containing negative powers of $\Lambda$. This circumstance allows us to correctly multiply such symbols by each other. Although it is not obvious, it turns out, as was shown  in \cite{Svinin3}, that the coefficients $S^j_s(n)$ in (\ref{678760}) are the very same polynomials defined by formulas (\ref{6000098652}) and (\ref{67500000000}). 

Let us also consider positive degrees
\[
\mathcal{T}^{s}=\Lambda^{-s}+\sum_{j= 1}^s  x^j T^j_s(n-s-j+1) \Lambda^{-s-j},\;\; s\geq 1.
\]
Here a new notation $T^j_s(n)$ appears for some  discrete polynomials to be defined. It was shown in \cite{Svinin3} that these discrete polynomials can be written explicitly as
\begin{equation}
T^k_s(n)=\sum_{0\leq \lambda_1< \cdots<  \lambda_k\leq s-1} \prod_{j=1}^k u_{n+\lambda_j+j-1}. 
\label{6897654000}
\end{equation}
It is convenient to assume that $T^k_s(n)=0$ for any $k\geq s+1$. 

Let us now define formal power series $\mathcal{S}_{s, n}(x)$ by the generating relation
\begin{equation}
\sum_{j=-\infty}^{\infty} \mathcal{T}^{j}= \sum_{j=-\infty}^{\infty} \mathcal{S}_{j, n}(x) \Lambda^j.
\label{68}
\end{equation}
In particular, by comparing the terms at different powers of $\Lambda$ in (\ref{68}), one can verify that 
\begin{equation}
\mathcal{S}_{s, n}(x)=1+\sum_{j\geq 1}(-1)^jS^j_{s+j}(n-j+1)x^j,\;\; \forall s\geq 0.
\label{680974530}
\end{equation}
These formal power series  play an important role in our further constructions.

\subsection{Technical lemmas}

Next, we are going to present  some identities that we will need later.    
\begin{lemma} \cite{Svinin3}
Polynomials $S^k_s(n)$ satisfy  two identities\footnote{When using this formulas, it is convenient to assume that  $S^k_0(n)=0$.} 
\begin{equation}
S^k_s(n)=S^k_{s-1}(n+1)+u_{n+k-1} S^{k-1}_{s}(n)
\label{6753098542}
\end{equation}
and
\begin{equation}
S^k_s(n)=S^k_{s-1}(n)+u_{n+s-1} S^{k-1}_{s}(n+1).
\label{675872}
\end{equation}
\end{lemma}
It should be noted that the recurrent relation (\ref{6000098652}) can be derived from (\ref{6753098542}). From (\ref{6753098542}) and (\ref{675872}) the following statement easily follows.
\begin{proposition}
The relations
\begin{equation}
\mathcal{S}_{s, n+1}=\mathcal{S}_{s+1, n}+u_{n}x \mathcal{S}_{s+2, n-1}\;\; \mbox{and}\;\; \mathcal{S}_{s, n}=\mathcal{S}_{s+1, n}+u_{n+s+1}x \mathcal{S}_{s+2, n},\;\;\forall s\geq 0,
\label{8997776}
\end{equation}
where the formal power series $\mathcal{S}_{s, n}$ is defined by (\ref{680974530}).
\end{proposition}
That is, this proposition can be proven by direct verification.
\begin{remark}
We assume that these two relations are also valid for any $s\leq -1$, but this already requires proof. 
\end{remark}
We will need a quite nontrivial property of the polynomials $S^k_s(n)$, formulated in the following lemma: 
\begin{lemma}
The relation
\begin{equation}
S^k_k(n)+S^k_k(n+1)=S^k_{k+1}(n),\;\; \forall  k\geq 1,
\label{8889554307}
\end{equation}
is an identity.
\end{lemma}
\begin{proof}
As a special case of (\ref{6753098542}) we have the relation 
\[
S^k_{k+1}(n)-S^k_{k}(n+1)=u_{n+k-1}S^{k-1}_{k+1}(n+1)
\]
that gives us the right to rewrite the relation (\ref{8889554307}) which we need to prove in an equivalent form as
\begin{equation}
S^k_k(n)=u_{n+k-1}S^{k-1}_{k+1}(n),
\label{887}
\end{equation}
but the latter is a special case of the factorization identity \cite{Svinin4}
\[
S^k_s(n)=S^{s-1}_{k+1}(n)\prod_{j=s-1}^{k-1} u_{n+j},\;\; s=1,\ldots, k
\]
holds.
\end{proof}
Note that, for $k=1$ and  $k=2$, the formula (\ref{8889554307}) was already used  at the end of the previous section. The following lemma is in order.
\begin{lemma}
The relation
\begin{equation}
2S^k_k(n)=S^k_{k+1}(n)+u_{n+k-1}S^{k-1}_k(n)-u_{n+k}S^{k-1}_k(n+2),\;\; \forall k\geq 1
\label{8880907}
\end{equation}
is an identity. 
\end{lemma}
\begin{proof}
By (\ref{8889554307}), we can rewrite (\ref{8880907}) as
\[
S^k_k(n)+u_{n+k}S^{k-1}_k(n+2)=S^k_{k}(n+1)+u_{n+k-1}S^{k-1}_k(n).
\]
By  (\ref{887}), we can rewrite the latter as 
\begin{equation}
u_{n+k-1}\left(S^{k-1}_{k+1}(n)-S^{k-1}_k(n)\right)=u_{n+k}\left(S^{k-1}_{k+1}(n+1)-S^{k-1}_k(n+2)\right).
\label{88007}
\end{equation}
Finally, using the identities 
\[
S^{k-1}_{k+1}(n)-S^{k-1}_k(n)=u_{n+k}S^{k-2}_{k+1}(n+1)
\]
and
\[
S^{k-1}_{k+1}(n+1)-S^{k-1}_k(n+2)=u_{n+k-1}S^{k-2}_{k+1}(n+1)
\]
that are particular cases of   (\ref{6753098542}) and (\ref{675872}), we see that (\ref{88007}) holds.
\end{proof}

\subsection{Stiltjes continued fraction}

Given an arbitrary  sequence $(u_n)$, let us consider  the corresponding Stiltjes continued fraction  
\begin{equation}
F_n(x)=1+\frac{u_n x}{1+\displaystyle{\frac{u_{n+ 1} x}{1+\displaystyle{\frac{u_{n+ 2} x}{1+\cdots}}}}}
\label{77643}
\end{equation}
and then define discrete polynomials $(h_{n, j})$ by the generating relation
\[
F_n(x)=1-\sum_{j\geq 1}(-1)^j h_{n, j} x^j.
\]
The first few coefficients $h_{n, j}$ look as
\[
h_{n, 1}=u_n,\;\; h_{n, 2}=u_nu_{n+ 1},\;\; h_{n, 3}=u_nu_{n+ 1}\left(u_{n+ 1}+u_{n+ 2}\right),\;\;
\]
\[
h_{n, 4}=u_nu_{n+ 1}\left((u_{n+ 1}+u_{n+ 2})^2+u_{n+2}u_{n+3}\right),\ldots
\]
The coefficient $h_{n, k}$, for any $k\geq 1$, is a polynomial in variables $\left(u_n,\ldots, u_{n+k-1}\right)$, but it is important to note that it depends linearly on the variable $u_{n+k-1}$.
This means that $u_{n+k-1}$, for any $k\geq 1$, can be expressed unambiguously as a rational function of the variables $\left(h_{n, 1},\ldots, h_{n, k}\right)$. 

By definition, the sequence $(F_n)$ satisfies the recurrence relation
\begin{equation}
F_nF_{n+ 1}=F_{n+ 1}+u_n x.
\label{85437}
\end{equation}
The following proposition relates the  polynomials $S^k_s(n)$  to the coefficients of the Stiltes continued fraction (\ref{77643}).
\begin{lemma}
The relation 
\begin{equation}
\mathcal{S}_{s, n}=F_{n+s+1}\mathcal{S}_{s+1, n},\;\; \forall s\geq 0
\label{80097}
\end{equation}
where the formal power series $\mathcal{S}_{s, n}$ is defined by (\ref{680974530}), holds.
\end{lemma}
\begin{proof}
Let us define a sequence  the formal power series $(F_{s, n})_{s\geq 0}$, by
\[
F_{s, n}=\frac{\mathcal{S}_{s, n}}{\mathcal{S}_{s+1, n}},\;\; \forall s\geq 0
\]
By virtue the second relation in (\ref{8997776}) we get the relation
\begin{equation}
F_{s, n}F_{s+1, n}=F_{s+1, n}+u_{n+s+1}x.
\label{80007}
\end{equation}
Let us define $h_{s, n, j}$ as coefficients in
\[
F_{s, n}(x)=1-\sum_{j\geq 1}(-1)^j h_{s, n, j} x^j.
\]
Substituting the latter into (\ref{80007}) we get, for example, $h_{s, n, 1}=u_{n+s+1}$. Thus $h_{s, n, 1}$ depends only on $n+s$. As is easy to see, the coefficients $h_{s, n, k}$ are determined using a recurrence relation
\begin{equation}
h_{s, n, k}=\sum_{j=1}^{k-1} h_{s, n, j} h_{s+1, n+1, k-j-1}
\label{8000987}
\end{equation}
which follows from (\ref{80007}). Using the relation (\ref{8000987}), by induction,  we obtain the fact that all coefficients $h_{s, n, k}$ depend only on  $n+s$. 
This in turn means that all members of the sequence of formal power series $(F_{s, n})_{s\geq 0}$ are actually expressed through a single series $F_n$ as $F_{s, n}=F_{n+s+1}$, that in turn is defined by (\ref{85437}). 
\end{proof}
\begin{remark}
Let us define `negative' Stiltjes continued fraction  
\[
F^{-}_n(x)=1+\frac{u_n x}{1+\displaystyle{\frac{u_{n- 1} x}{1+\displaystyle{\frac{u_{n- 2} x}{1+\cdots}}}}}
\]
It is  obvious that it satisfies the recurrence relation
\[
F^{-}_nF^{-}_{n-1}=F^{-}_{n-1}+u_n x.
\]
But what is not obvious is that the relation 
\[
\mathcal{S}_{s, n+1}=F^{-}_{n}\mathcal{S}_{s+1, n},\;\; \forall s\geq 0
\]
similar to the relation (\ref{80097})  holds. Note that we will not need the last relation further.
\end{remark}

\subsection{Triple of formal series} 

In the construction presented below we will need series of the form (\ref{680974530}). Namely, let us  define the triple of formal power series $(\mathcal{P}_n(x), \mathcal{Q}_n(x), \mathcal{R}_n(x))$ by
\begin{eqnarray}
\mathcal{P}_n&=&2 \mathcal{S}_{0, n-1}-1 \label{099991} \\
&=&1-2 u_{n-1}x \mathcal{S}_{2, n-2}, \label{0956421}
\end{eqnarray}
\begin{equation}
\mathcal{Q}_n=2 \mathcal{S}_{1, n-1}\;\; \mbox{and}\;\; \mathcal{R}_n=2 u_{n-1} x \mathcal{S}_{1, n-2}.
\label{90834521}
\end{equation}
Note that $\mathcal{P}_n$ is defined here in two ways due to the identity 
\begin{equation}
1-u_{n+1}x \mathcal{S}_{2, n}=\mathcal{S}_{0, n+1}
\label{9000}
\end{equation}
that is a consequence of the factorization property (\ref{887}).
\begin{lemma} \label{77765}
Three linear relations 
\begin{equation}
\mathcal{P}_n+\mathcal{P}_{n+1}=\mathcal{Q}_n,\;\; \mathcal{R}_{n+1}=u_n x \mathcal{Q}_n,\;\; 2 \mathcal{P}_n - \mathcal{Q}_n + \mathcal{R}_n=u_n x \mathcal{Q}_{n+1}
\label{98342}
\end{equation}
and quadratic  one
\begin{equation}
\mathcal{P}^2_n+\mathcal{Q}_n\mathcal{R}_n=1,
\label{98765200}
\end{equation}
for the triple $(\mathcal{P}_n, \mathcal{Q}_n, \mathcal{R}_n)$ given by (\ref{0956421}) and (\ref{90834521}), are identities.
\end{lemma}
\begin{proof}
The first relation in (\ref{98342}) follows from the identity (\ref{8889554307}), while the second one follows from the definition of these formal power  series. Finally, the third relation  in   (\ref{98342}), perhaps the most non-trivial, can be obtained from the identity (\ref{8880907}). It remains to prove that the relation (\ref{98765200}) is an identity. By (\ref{0956421}) and (\ref{90834521}), we have
\begin{eqnarray*}
\mathcal{P}^2_n+\mathcal{Q}_n\mathcal{R}_n&=& \left(1-2u_{n-1}x \mathcal{S}_{2, n-2}\right)^2 +4u_{n-1}x \mathcal{S}_{1, n-2} \mathcal{S}_{1, n-1}  \\
&=&1-4u_{n-1}x \mathcal{S}_{2, n-2}+4u^2_{n-1}x^2 \mathcal{S}_{2, n-2}^2  +4u_{n-1}x \mathcal{S}_{1, n-2}\mathcal{S}_{1, n-1}.
\end{eqnarray*}
From this we see that  result follows provided the relation  $\mathcal{S}_{2, n}\left(1-u_{n+1}x \mathcal{S}_{2, n}\right)=\mathcal{S}_{1, n}\mathcal{S}_{1, n+1}$ is an identity. Using   (\ref{9000}), we can rewrite it  as  
\begin{equation}
\mathcal{S}_{2, n}\mathcal{S}_{0, n+1}=\mathcal{S}_{1, n}\mathcal{S}_{1, n+1}
\label{987453}
\end{equation}
but the latter can be  proved using  (\ref{80097}). Indeed, substituting $\mathcal{S}_{1, n}=F_{n+2} \mathcal{S}_{2, n}$ into the right-hand side of the relation (\ref{987453}) leads it to $\mathcal{S}_{0, n+1}=F_{n+2} \mathcal{S}_{1, n+1}$ which is  a special case of (\ref{80097}).
\end{proof}

It should be noted that the identities (\ref{98342})  can be written in the form of a matrix relation which formally looks like the Lax equation 
$\mathcal{L}_n\mathcal{M}_n=\mathcal{M}_n\mathcal{L}_{n+1}$, where
\[
\mathcal{L}_n=
\left(
\begin{array}{cc}
\mathcal{P}_n & \mathcal{R}_n \\
\mathcal{Q}_n & -\mathcal{P}_n 
\end{array}
\right)\;\; \mbox{and}\;\;
\mathcal{M}_n= 
\left(
\begin{array}{cc}
1 & u_n x \\
1 & 0 
\end{array}
\right),
\]
however, it is important to note that  this relation does not define any dynamical system since the sequence $(u_n)$ remains free in this framework. From  Lemma \ref{77765}, it follows that $\det \mathcal{L}_n=1$. Note also that due to (\ref{80097}) and  (\ref{099991}) we get
\begin{equation}
F_n=\frac{2 \mathcal{S}_{0, n-1}}{2 \mathcal{S}_{1, n-1}}=\frac{1+\mathcal{P}_n}{\mathcal{Q}_n}.
\label{9874430}
\end{equation}

Let $y(x)=1+\sum_{j\geq 1}H_jx^j$, where $H_j$'s are supposed to be arbitrary constants and
\begin{equation}
P_n=\mathcal{P}_n y,\;\; Q_n=\mathcal{Q}_n y,\;\; R_n=\mathcal{R}_n y.
\label{8965421098}
\end{equation}
It is obvious that the triple $(P_n, Q_n, R_n)$ satisfies the same linear homogeneous identities (\ref{98342}), while (\ref{98765200}) becomes
$y^2=P^2_n+Q_nR_n$. In turn the relation (\ref{9874430}) takes the form $F_n=\left(y+P_n\right)/Q_n$.

So we have a triple of formal series
\begin{equation}
P_n(x)=1+\sum_{j\geq 1}p_{n, j}x^j,\;\; Q_n(x)=2+\sum_{j\geq 1}q_{n, j}x^j,\;\;   R_n(x)=\sum_{j\geq 1} r_{n, j}x^j
\label{66543}
\end{equation}
with coefficients that have the following form:
\begin{eqnarray}
p_{n, k}&=&(-1)^k 2\sum_{j=0}^{k-1}(-1)^jS_{k-j}^{k-j}(n+j-k)H_j+H_k \label{98005} \\ 
&=&(-1)^k 2 u_{n-1} \hat{S}^{k-1}_{k+1}(n-k)+H_k,
\label{985325}
\end{eqnarray}
\begin{equation}
q_{n, k}=p_{n, k}+p_{n+1, k}=(-1)^k 2 \hat{S}^{k}_{k+1}(n-k),
\label{887453}
\end{equation}
\begin{equation}
r_{n, k}= u_{n-1} q_{n-1, k-1}= (-1)^{k-1} 2  u_{n-1} \hat{S}^{k-1}_{k}(n-k),\;\; \forall k\geq 1.
\label{88000}
\end{equation}
What is important for us and what we actually achieved is that these formal power series satisfy the relations (\ref{8877545}), but while the triple $(P_n, Q_n, R_n)$ is represented by formal power series, these relations are only identities and do not define any recurrences. In the next section, we show how to use a suitable truncation of a formal power series $P_n(x)$ to obtain a Lax representation for a recurrence of order $2g+1$. 
\begin{remark}
Let us say a few words to clarify some details. In the formulas (\ref{985325}), (\ref{887453}) and  (\ref{88000}), we have used the convenient notation
\begin{equation}
\hat{S}^k_s(n)=\sum_{j=0}^{k} (-1)^j  S^{k-j}_{s-j}(n+j)H_j
\label{8909876}
\end{equation}
that we will continue to use throughout the text. For convenience, it is assumed here that $H_0=1$. 
It is easy to verify that  expression (\ref{98005}) is equivalent to (\ref{985325}) due to the factorization  identity (\ref{887}). In turn, to obtain (\ref{887453}) we have used the identity (\ref{8889554307}). 
\end{remark}

\section{Recurrence of order $2g+1$}
\label{6752109654}

The goal of this section is to write explicit formulas that relate the system of equations (\ref{880099875})  to some  polynomial recurrence on the sequence $(u_n)$  of order $2g+1$.    In the previous section, we actually prepared the ground for solving this problem.

\subsection{Main theorem}

Let us consider a triple of infinite formal power series $(P_n, Q_n, R_n)$ defined by (\ref{8965421098}). Their coefficients are determined by  (\ref{985325}), (\ref{887453}) and (\ref{88000}). 
Let us suppose  that 
\begin{equation}
p_{n, k}=0,\;\; \forall k\geq g+1
\label{094390}
\end{equation}
for some $g\geq 1$.  We observe that, by virtue of (\ref{887453}) and (\ref{88000}), truncation only the series $P_n(x)$ in fact entails truncation  the series $Q_n(x)$ and $R_n(x)$ in such a way that they become polynomials in $x$ of the form (\ref{87324}).  As a result, the restriction (\ref{094390}) forces the sequence $(u_n)$ to  be a solution to some equation, which we will write out below. Indeed, from the second and third relation in (\ref{880099875}), it follows that 
\[
q_{n+1, g}=\frac{r_{n, g+1}}{u_n}\;\; \mbox{and}\;\; r_{n+1, g+1}=u_nq_{n, g}
\]
and as a consequence we get the relation
\begin{equation}
u_{n+g}q_{n+g, g}=u_{n+g+1}q_{n+g+2, g}.
\label{786532098}
\end{equation}
Finally, taking into account (\ref{887453}),  we can rewrite (\ref{786532098}) as the equation 
\begin{equation}
u_{n+g}\hat{S}^g_{g+1}(n)=u_{n+g+1}\hat{S}^g_{g+1}(n+2),
\label{88711111}
\end{equation}
where, let us remind ourselves, we use the notation (\ref{8909876}). This is actually the polynomial recursion (\ref{98430}) we wanted to derive.  Next we will show that (\ref{88711111}) represents a known equation and of course some facts are known for it.

If we look at the explicit expression (\ref{67500000000}) for the polynomials $S^k_s(n)$, we see that the relation (\ref{88711111}) is a discrete equation of order $2g+1$. Fortunately, the variable $u_{n+2g+1}$ enters the relation (\ref{88711111}) linearly and thus we can, if necessary, resolve this relation with respect to this variable to rewrite it as 
\[
u_{n+2g+1}=R(u_n,\ldots, u_{n+2g}; H_1,\ldots, H_g),
\]
where $R$ is some rational function of its arguments although it is clear that in this form this equation does not look so nice as (\ref{88711111}) does. Any sequence satisfying (\ref{88711111}) is uniquely determined by a finite set $(u_0,\ldots, u_{2g}, H_1,\ldots, H_g)\in\mathbb{C}^{3g+1}$, that is, the number of parameters required to uniquely determine the solution in this case is equal to $3g+1$, that is, this number is equal to the dimension of the phase space $M_g$ on which the Volterra map is defined. 
The difference between desired  recursion (\ref{98430})  and the relation (\ref{88711111}) is only that there are parameters $\left(c_1,\ldots, c_g\right)$ in (\ref{98430}), however in (\ref{88711111})  one sees the parameters  $\left(H_1,\ldots, H_g\right)$. But this is not a problem, because, as we already said, these parameters are related to each other by a birational transformation (\ref{00985463}).

Let us  formulate the main result of this paper in the form of the following theorem:
\begin{theorem} \label{8887}
Given any solution of the system of equations (\ref{880099875})  generated by the Volterra map $\mathcal{V}_g$, the sequence  $(u_n)$, where $u_n=p_{n, 1}-q_{n, 1}/2+r_{n, 1}/2$ satisfies  equation (\ref{88711111}) with  parameters $\left(H_1,\ldots, H_g\right)$ related birationally to the invariants $\left(c_1,\ldots, c_g\right)$ by the generating relation
\begin{equation}
\left(1+\sum_{j\geq 1} H_j x^j\right)^2=1+\sum_{j\geq 1}^{2g+1} c_j x^j.
\label{09342}
\end{equation}
\end{theorem}

\subsection{Birational correspondence}

In addition to  Theorem \ref{8887}, we note the following. Formulas (\ref{985325}), (\ref{887453}) and (\ref{88000}) provide us explicit expressions for the polynomial mapping $\left(u_{-g-1},\ldots, u_{g-1}; H_1,\ldots, H_g\right) \rightarrow \left(p_1,\ldots, p_g, q_1,\ldots, q_g, r_1,\ldots, r_{g+1}\right)$. Theoretically,  one can obtain formulas for rational mapping $\left(u_0,\ldots, u_{2g}; H_1,\ldots, H_g\right) \rightarrow \left(p_1,\ldots, p_g, q_1,\ldots, q_g, r_1,\ldots, r_{g+1}\right)$ but it looks more cumbersome.

To illustrate the above, let us  give two  examples.
\begin{example} $g=1$. In this case, equation (\ref{88711111}) is specified as the third-order equation
\begin{equation}
u_{n+1}\left(u_n+u_{n+1}-H_1\right)=u_{n+2}\left(u_{n+2}+u_{n+3}-H_1\right).
\label{16753090097438}
\end{equation}
In this case  the  polynomial mapping $\left(u_{-2}, u_{-1}, u_0, H_1\right)\rightarrow \left(p_1, q_1, r_1, r_2\right)$ is defined by
\begin{equation}
p_1=-2u_{-1}+H_1,\;\; q_1=-2\left(u_{-1}+u_0-H_1\right),\;\; r_1=2u_{-1},\;\; r_2=-2u_{-1}\left(u_{-2}+u_{-1}-H_1\right).
\label{90632453}
\end{equation}
The inverse mapping   looks like this:
\begin{equation}
u_{-2}=\frac{2p_1+r_1}{2}-\frac{r_2}{r_1},\;\; u_{-1}=\frac{r_1}{2},\;\; u_{0}=\frac{2 p_1-q_1+r_1}{2},\;\; H_1=p_1+r_1.
\label{8973209}
\end{equation}
\end{example}
\begin{example} $g=2$. 
In this example, to avoid cumbersomeness, we use the notation of discrete polynomials $S^k_s(n)$.
In this case, equation (\ref{88711111}) is specified as the fifth-order equation
\begin{equation}
u_{n+2}\left(S^2_{3}(n)-H_1 S^1_{2}(n+1)+H_2\right)=u_{n+3}\left(S^2_{3}(n+2)-H_1 S^1_{2}(n+3)+H_2\right).
\label{16888658}
\end{equation}
The  polynomial mapping $\left(u_{-3}, u_{-2}, u_{-1}, u_0, u_1, H_1, H_2\right)\rightarrow \left(p_1, p_2, q_1, q_2, r_1, r_2, r_3\right)$ is defined by (\ref{90632453}) together with
\[
p_2=2u_{-1}\left(S^1_3(-2)-H_1\right)+H_2,\;\; 
q_2=2\left(S^2_3(-2)-H_1 S^1_2(-1)+H_2\right),
\]
\[
r_3=2 u_{-1} \left(S^2_3(-3)-H_1 S^1_2(-2)+H_2\right).
\]
In turn, the inverse mapping is defined by  (\ref{8973209}) supplemented with 
\[
u_{-3}=\frac{4p_1r_1r_2+2p_1r_1^3+4r_1r_3+r_1^4-q_1r_1^3-4p_2r_1^2-4r_2^2}{2r_1\left(2p_1r_1+r_1^2-2r_2\right)},
\]
\[
u_1=\frac{q_1+r_1}{2}-\frac{2p_2+r_2-q_2}{2p_1+r_1-q_1}
\]
and
\[
H_2=p_2+r_2-p_1r_1+\frac{q_1 r_1}{2}-\frac{r_1^2}{2}.
\]
\end{example}
\begin{remark}
Note  that  (\ref{16888658}) was written out explicitly in \cite{Hone1} as the Equation (3.51), but it looks there  in a rather cumbersome form.  
\end{remark}

\subsection{Invariants for  equation (\ref{88711111})} 

Obviously, that by virtue of (\ref{985325}),   we can rewrite  (\ref{094390}) as 
\begin{equation}
H_k=(-1)^{k+1} 2 u_{n-1} \hat{S}^{k-1}_{k+1}(n-k),\;\; \forall k\geq g+1.
\label{1688}
\end{equation}
It is important to note that we have the right to impose   such restrictions only   if the right-hand side of (\ref{1688}), by virtue of the equation (\ref{88711111}) is independent of $n$ or in other words is an invariant for this equation. If this is indeed the case, then we may rewrite (\ref{1688}) as  
\begin{equation}
H_k=(-1)^{k+1} 2 u_{n+k-1} \hat{S}^{k-1}_{k+1}(n),\;\; \forall k\geq g+1.
\label{168899987}
\end{equation}
Let us take  $k=g+1$ in (\ref{168899987}) first. We denote
\begin{equation}
\mathcal{G}_{g, 0}=u_{n+g} \hat{S}^g_{g+2}(n) 
\label{987887} 
\end{equation}
Next we will prove that (\ref{987887}) is indeed an invariant for  the equation (\ref{88711111}) in a somewhat broader context using the ready results of the work \cite{Svinin4} but for now let us take it on faith. 
\begin{remark}
Observe that the right-hand side of (\ref{987887}) is some polynomial in a finite number of variables $\left(u_{n},\ldots, u_{n+2g}\right)$ and therefore we can consider this relation  as an equation of order $2g$. 
In particular, in the case $g=2$,  (\ref{987887}) coincides with the equation (\ref{983210987}) in the Introduction. 
\end{remark}

It turns out that is not at all obvious that to determine a larger number of invariants for  equation (\ref{88711111}), we need discrete polynomials
\begin{equation}
Q_s^k(n)=S_s^k(n)-u_{n+k-2}u_{n+s}S_{s+2}^{k-2}(n).
\label{00000932}
\end{equation}
For instance,  $Q_s^1(n)=S_s^1(n)\;\; \mbox{and}\;\; Q_s^2(n)=S_s^2(n)-u_{n}u_{n+s},\; \forall s\geq 1$. 

For the following proposition we will present further  a proof.
\begin{proposition} \label{983421098}
Given any $g\geq 1$, a discrete polynomial 
\[
\mathcal{G}_{g, r}=\sum_{j=0}^r Q^j_{2r-j}(n+g-r+1) G_{g, r-j}(n),\;\; \forall r=0,\ldots, g,
\]
where
\begin{equation}
G_{g, r}(n)=\prod_{j=g-r}^{g+r}u_{n+j} \hat{S}^{g-r}_{g+r+2}(n),
\label{9532}
\end{equation}
is an invariant for the equation (\ref{88711111}).
\end{proposition}

Now we need to prove that,  under these circumstances, the  $H_k$, for any $k\geq g+2$,  are also invariant for the equation (\ref{88711111}).  The following lemma will be useful for this.
\begin{lemma}
Any solution of  $g$-th equation (\ref{88711111}) is also the solution $(g+1)$-th equation (\ref{88711111}) with suitable parameters.
\end{lemma}
\begin{proof} 
To prove this proposition, we use the birational relationship of equation (\ref{88711111}) with the system (\ref{880099875}). 
Any solution of the $g$-th equation (\ref{88711111}), due to birational correspondence, is associated to some solution of the $g$-th system of equations (\ref{880099875}). But on the other hand, a corresponding solution of the $g$-th system is also a some solution of the $(g+1)$-th system of equations (\ref{880099875}), with obvious constraints: $p_{n, g+1}=q_{n, g+1}=r_{n, g+2}=0$. Next, with the help of a birational transformation, a solution of the $(g+1)$-th system of equations (\ref{880099875}) is mapped into a some solution of  the $(g+1)$-th equation (\ref{88711111}) and thus the lemma is proven.
\end{proof} 
As a consequence of this lemma, we obtain
\begin{proposition}
\label{095430}
Any solution of the $g$-th equation (\ref{88711111}) is also the solution $(g+r)$-th equation (\ref{88711111}), with any $r\geq 1$, with suitable parameters.
\end{proposition}
Above we have announced a proof that $\mathcal{G}_{g, 0}$ given by (\ref{987887})  is an invariant for the  equation (\ref{88711111}) for any fixed $g\geq 1$. Assuming the validity of this statement and also taking into account  Proposition \ref{095430},  we obtain that $u_{n+g+r} \hat{S}^{g+r}_{g+r+2}(n),\;\; \forall r\geq 1$ is an invariant for the  equation (\ref{88711111}) and from this, in turn, it follows that  the constraints  (\ref{168899987}) are defined correctly. However, the following question remains unclear. The fact is that $u_{n+g+r} \hat{S}^{g+r}_{g+r+2}(n)$ depends on the variables $\left(u_n,\ldots, u_{n+2g+2r}\right)$ and therefore, generally speaking, $u_{n+g+r} \hat{S}^{g+r}_{g+r+2}(n)$ should be brought to such a form  it depends on the variables $\left(u_n,\ldots, u_{n+2g}\right)$, that is, it is necessary to express variables $\left(u_{n+2g+1},\ldots, u_{n+2g+2r}\right)$ using (\ref{88711111}). But what do we get as a result? It is clear that for this equation there cannot exist an infinite number of independent invariants. We will give examples below, but for now it is appropriate to substantiate the above statements.

\subsection{Two-parameter class of discrete equations. Invariants} 

First of all, it should be noted that  (\ref{88711111}) actually represent a special case of equations of the form 
\begin{equation}
u_{n+k}\hat{S}^k_s(n)=u_{n+s}\hat{S}^k_s(n+2),
\label{16753098}
\end{equation}
where, it is supposed that $k\geq 1$ and $s\geq k+1$. This two-parameter class of recurrences have been considered in  \cite{Svinin2}, \cite{Svinin3}, \cite{Svinin4}. Note that the variable $u_{n+s+k}$ enters the right-hand side of the relation (\ref{16753098}) linearly and therefore we can resolve (\ref{16753098}) with respect to this variable. As a result we obtain this equation in the form 
\[
u_{n+N}=R(u_n,\ldots, u_{n+N-1}; H_1,\ldots, H_k),
\]
where $N=s+k$, is obviously  the order of the equation (\ref{16753098}).

In fact, these equations arose in the above mentioned works as constraints compatible with the Volterra lattice equation and its hierarchy of generalized symmetries (\ref{977657}). 
Now, we would like to show how a certain set of invariants can be constructed for the equation (\ref{16753098}). It is natural that we are mainly interested in the case $(k, s)=(g, g+1)$, where $g\geq 1$. 

To begin with, we present the following lemma.
\begin{lemma} 
\cite{Svinin4}
Equation (\ref{16753098}) can be written in the following equivalent form:
\begin{equation}
u_{n+k}\hat{S}^k_{s+1}(n)=u_{n+s}\hat{S}^k_{s+1}(n+1).
\label{200000}
\end{equation}
\end{lemma}
From this lemma it clearly follows that the discrete polynomial 
\begin{equation}
G_0^{(k, s)}(n)=\hat{S}^k_{s+1}(n)\prod_{j=k}^{s-1}u_{n+j}
\label{90432983}
\end{equation}
is an invariant for the equation (\ref{16753098}). In particular, if $(k, s)=(g, g+1)$, then we obtain the fact that $G_0^{(g, g+1)}=\mathcal{G}_{g, 0}$ is an invariant for the equation (\ref{88711111}). 

Now we would like to show the proof of the Proposition \ref{983421098} in a broader context. Let us denote
\begin{equation}
G_r^{(k, s)}(n)=\sum_{j=0}^r Q^j_{s-k-j-1}(n+k+1) G_0^{(k+j, s-j)}(n),\;\; r=0,\ldots, k,
\label{2777000987}
\end{equation}
where the polynomials $Q^k_s(n)$ are defined by (\ref{00000932}).
\begin{lemma} \cite{Svinin4}
The relation
\begin{equation}
G_r^{(k, s)}(n+1)-G_r^{(k, s)}(n)=\Lambda_r^{(k, s)}(n)\left(u_{n+s-r}\hat{S}_{s-r+1}^{k+r}(n+1) -u_{n+k+r}\hat{S}_{s-r+1}^{k+r}(n)\right) 
\label{38976}
\end{equation}
is an identity, where the integrating factor looks as
\[
\Lambda_r^{(k, s)}(n)=S^r_{s-k-r}(n+k+1) \prod_{j=k+r+1}^{s-r-1}u_{n+j}.
\]
\end{lemma}
From this lemma we can easily draw the following conclusion.
\begin{proposition} \label{908432098}
Given a pair $(k, s)$ such $k\geq 1$ and $s\geq k+1$, by   (\ref{38976}), the discrete polynomial 
\begin{equation}
\mathcal{G}_r^{(k, s)}=G_r^{(k-r, s+r)},\;\; r=0,\ldots, k,
\label{3897600098}
\end{equation}
where $G_r^{(k, s)}(n)$ is calculated with the help of (\ref{2777000987}), is an invariant for the recurrence (\ref{16753098}).
\end{proposition}
In the particular case $(k, s)=(g, g+1)$, taking into account  (\ref{2777000987}) and (\ref{3897600098}), ultimately, we get a proof of the Proposition \ref{983421098}, where $\mathcal{G}_{g, r}=\mathcal{G}_r^{(g, g+1)}=G_r^{(g-r, g+r+1)}$.
\begin{remark}
It can be assumed, and calculations confirm this fact, that the set of invariants $(\mathcal{G}_0^{(k, s)},\ldots, \mathcal{G}_k^{(k, s)})$ given by (\ref{3897600098}) satisfies the triangular linear system of equations
\begin{equation}
\sum_{j=0}^r(-1)^j T^j_{s-k-j-1}(n+k+1) \mathcal{G}_j^{(k, s)}=G_0^{(k-r, s+r)},\;\; r=0,\ldots, k,
\label{300000000}
\end{equation}
where discrete polynomials $T^k_s(n)$ are defined by  (\ref{6897654000}), where the discrete polynomials $G_0^{(k, s)}$ are given by (\ref{90432983}).  It should be noted that, in principle, we could prove (\ref{300000000}) if we  prove the identity 
\[
\sum_{j=0}^k (-1)^j Q^{k-j}_{s-j}(n+j)T^j_{s+k-j}(n)=0
\]
for any $k\geq 1$ and $s\geq k$. For $k=1$ this identity is trivial, while for two cases $k=2$ and $k=3$ it has been proved in the work \cite{Svinin4}. However, it should be noted that the lack of proof of this statement does not prevent us from verifying this fact for any specific case.
\end{remark}

\subsection{Examples and conjecture} 

Let us consider now two examples of calculations of invariants $\mathcal{G}_{g, r}$ defined in the Proposition \ref{983421098}.
\begin{example} \label{9084532} $g=1$. We already know that, in this case, equation (\ref{88711111}) is specified as the third-order equation (\ref{16753090097438}). The Proposition \ref{983421098}  gives us  the following two invariants:
\begin{equation}
\mathcal{G}_{1, 0}=u_{1}\left(u_0+u_{1}+u_{2}-H_1\right)\;\; \mbox{and}\;\; \mathcal{G}_{1, 1}=u_0u_{1}u_{2}+u_{1}^2\left(u_0+u_{1}+u_{2}-H_1\right).
\label{43296665325}
\end{equation}
These  can be determined from the linear relations 
\[
\mathcal{G}_{1, 0}=u_{1}\hat{S}^1_3(0)\;\; \mbox{and}\;\; \mathcal{G}_{1, 1}-u_1\mathcal{G}_{1, 0}=u_0u_{1}u_{2}.
\]
Although these relations are obvious, nevertheless, it is important to note that that they are a special case of (\ref{300000000}). 
Actual calculations gives 
\[
H_2=-2\mathcal{G}_{1, 0},\;\; H_3=2\mathcal{G}_{1, 1},\;\; H_4=-2 H_1 \mathcal{G}_{1, 1}-2 \mathcal{G}_{1, 0}^2,
\]
\[
H_5=2 H_1^2 \mathcal{G}_{1, 1} + 2 H_1 \mathcal{G}_{1, 0}^2+ 4 \mathcal{G}_{1, 0} \mathcal{G}_{1, 1},
\]
\begin{equation}
H_6=-2H_1^3\mathcal{G}_{1, 1}-8 H_1 \mathcal{G}_{1, 0}  \mathcal{G}_{1, 1}-2H_1^2\mathcal{G}_{1, 0}^2 -4 \mathcal{G}_{1, 0}^3-2\mathcal{G}_{1, 1}^2,\ldots
\label{432986008545}
\end{equation}
In this case,  we have the following three nonzero $c_k$'s:
\[
c_1=2H_1,\;\; c_2=-4 \mathcal{G}_{1, 0} + H_1^2,\;\; c_3 = 4 \mathcal{G}_{1, 1} -4 H_1 \mathcal{G}_{1, 0}.
\]
\end{example} 
\begin{example} $g=2$. In this case, the equation (\ref{88711111}) is specified as the fifth-order equation (\ref{16888658}). The Proposition \ref{983421098}  gives  the following three invariants: 
\begin{equation}
\mathcal{G}_{2, 0}=u_{2}\left(S^2_4(0)-H_1S^1_3(1)+H_2\right),
\label{84328887435}
\end{equation}
\begin{equation}
\mathcal{G}_{2, 1}=u_{1}u_{2}u_{3}\left(S^1_5(0)-H_1\right)+u_{2}^2\left(S^2_4(0)-H_1S^1_3(1)+H_2\right)
\label{8432866543875}
\end{equation}
and
\begin{eqnarray}
\mathcal{G}_{2, 2}&=&u_{0}u_{1}u_{2}u_{3}u_{4}+\left(u_{1}+u_{2}+u_{3}\right)u_{1}u_{2}u_{3}\left(S^1_5(0)-H_1\right) \nonumber \\
&&+\left(u_{2}(u_{1}+u_{2})+u_{3}u_{2}-u_{1}u_{3}\right)u_{2}\left(S^2_4(0)-H_1S^1_3(1)+H_2\right).
\label{11215432876}
\end{eqnarray}
These invariants are determined from the linear relations
\[
\mathcal{G}_{2, 0}=u_2 \hat{S}^2_4(0),\;\; \mathcal{G}_{2, 1}-u_2 \mathcal{G}_{2, 0}=u_1u_2u_3 \hat{S}^1_5(0),\;\; 
\]
and
\[
\mathcal{G}_{2, 2}-\left(u_1+u_2+u_3\right) \mathcal{G}_{2, 1}+u_1u_3 \mathcal{G}_{2, 0} = u_0u_1u_2u_3u_4. 
\]

Actual calculations gives 
\[
H_3=2\mathcal{G}_{2, 0},\;\; H_4=-2\mathcal{G}_{2, 1},\;\;  H_5=2\mathcal{G}_{2, 2},\;\; H_6=2H_2\mathcal{G}_{2, 1}-2 H_1\mathcal{G}_{2, 2}-2\mathcal{G}_{2, 0}^2,
\]
\begin{equation}
H_7=-2H_1H_2\mathcal{G}_{2, 1}+2\left(H_1^2-H_2\right)\mathcal{G}_{2, 2}+2H_1 \mathcal{G}_{2, 0}^2+ 4\mathcal{G}_{2, 0}\mathcal{G}_{2, 1},\ldots
\label{4329865}
\end{equation}
In this case, we have the following five nonzero $c_k$'s:
\[
c_1=2 H_1,\;\; c_2=2 H_2 + H_1^2,\;\; c_3=4 \mathcal{G}_{2, 0} + 2 H_1 H_2,\;\; c_4= -4 \mathcal{G}_{2, 1} +4 H_1 \mathcal{G}_{2, 0} + H_2^2
\]
\[
c_5= 4 \mathcal{G}_{2, 2}-4 H_1 \mathcal{G}_{2, 1}+4 H_2 \mathcal{G}_{2, 0}.
\]
\end{example}
Inspired by these two examples based on actual calculations, we can assume the following statement could be valid.
\begin{conjecture}
Given any $g\geq 1$,   in virtue of equation (\ref{88711111}) we have
\begin{equation}
H_{g+k+1}=(-1)^{g+k} 2  \mathcal{G}_{g, k},\;\; k=0,\ldots, g,
\label{78632098}
\end{equation}
where $H_k,\;\; k\geq g+1 $ is given by (\ref{168899987}). For  the remaining  values $k\geq 2g+2$, the coefficients $H_k$  are some polynomials in the variables $\left(H_1,\ldots, H_g; \mathcal{G}_{g, 0},\ldots, \mathcal{G}_{g, g}\right)$. 
\end{conjecture}  

\section{Homogeneous recurrences. Laurent property. Integer sequences}
\label{99832098}

In essence, the goal of this section is to show with examples how some substitution reduces the  equation (\ref{88711111}) written above to homogeneous equations sharing the Laurent property.

\subsection{Homogeneous recurrences}

Let us consider a substitution
\begin{equation}
u_n=\frac{t_nt_{n+3}}{t_{n+1}t_{n+2}}.
\label{896432097}
\end{equation}
We formulate below three lemmas that give a homogeneous recursion relations of order $2g+3$ for $g=1, 2, 3$.  All these statements can be proven by direct verification.
\begin{lemma} $g=1$.
Provided (\ref{896432097}), the relation
\begin{equation}
t_nt_{n+5}+\mathcal{G}_{1, 0} t_{n+1}t_{n+4}-\mathcal{G}_{1, 1} t_{n+2}t_{n+3}=0,
\label{896432000096}
\end{equation}
where $\mathcal{G}_{1, 0}$ and  $\mathcal{G}_{1, 1}$ are given by  (\ref{43296665325}) is an identity.
\end{lemma}
\begin{lemma} \label{6549098} $g=2$.
Provided (\ref{896432097}), the relation
\begin{equation}
t_{n+3}t_{n+4}\left(t_nt_{n+7}-\mathcal{G}_{2, 0} t_{n+1}t_{n+6}-\mathcal{G}_{2, 2} t_{n+2}t_{n+5}\right)=-\mathcal{G}_{2, 1} B_{n+1}, 
\label{87621112342}
\end{equation}
where $\mathcal{G}_{2, 0}$, $\mathcal{G}_{2, 1}$ and $\mathcal{G}_{2, 2}$ are given by  (\ref{84328887435}), (\ref{8432866543875}) and (\ref{11215432876}), respectively, while 
\begin{equation}
B_n=t_nt_{n+3}^2t_{n+4}+t_{n+1}t_{n+2}^2t_{n+5}+t_{n+1}^2t_{n+4}^2
\label{87766}
\end{equation}
is an identity.
\end{lemma}
It is easy to see that the polynomial $B_n$ given by (\ref{87766}) has such a symmetry that  the relation (\ref{87621112342}) is invariant under the  reversing
$\left(t_0, t_{1}, t_{2}, t_{3}, t_{4}, t_{5}, t_{6}, t_{7}\right)\rightarrow \left(t_{7}, t_{6}, t_{5}, t_{4}, t_{3}, t_{2}, t_{1}, t_0\right)$.
\begin{lemma} 
\label{654309865}
$g=3$.
Provided (\ref{896432097}), the relation
\begin{equation}
t_{n+3}t_{n+4}t_{n+5}t_{n+6}\left(t_nt_{n+9}+\mathcal{G}_{3, 0} t_{n+1}t_{n+8}-\mathcal{G}_{3, 3} t_{n+2}t_{n+7}\right)=\mathcal{G}_{3, 1} B_{n+1, 1}  -  \mathcal{G}_{3, 2} B_{n+1, 2}
\label{00432000985436}
\end{equation}
is an identity, where
\begin{eqnarray}
B_{n, 1}&=&t_nt_{n+3}^3t_{n+6}^2+t_{n+1}^2t_{n+4}^3t_{n+7}+t_nt_{n+3}^2t_{n+4}^2t_{n+7}+t_{n+1}^2t_{n+3}t_{n+4}t_{n+6}^2 \nonumber\\
&&+t_nt_{n+2}t_{n+3}t_{n+5}^2t_{n+6}+t_{n+1}t_{n+2}^2t_{n+4}t_{n+5}t_{n+7}
\label{896432000985436}
\end{eqnarray}
and
\begin{eqnarray}
B_{n, 2}&=&t_nt_{n+3}^2t_{n+4}t_{n+5}t_{n+6}+t_{n+1}t_{n+2}t_{n+3}t_{n+4}^2t_{n+7} \nonumber\\
&&+t_{n+1}t_{n+2}t_{n+3}^2t_{n+6}^2+t_{n+1}t_{n+2}^2t_{n+5}^2t_{n+6}+t_{n+1}^2t_{n+4}^2t_{n+5}t_{n+6}.
\label{89677766}
\end{eqnarray}
\end{lemma}
\begin{remark}
We have not write out the invariants $\left(\mathcal{G}_{3, 0}, \mathcal{G}_{3, 1}, \mathcal{G}_{3, 2}, \mathcal{G}_{3, 3}\right)$ in the examples explicitly, due to their bulkiness, but be that as it may, they are defined above and therefore  they can be calculated to verify Lemma \ref{654309865}.
\end{remark}
\begin{remark}
A similar remark to that made about the symmetry of $B_n$  can be made regarding  polynomials $B_{n, 1}$ and $B_{n, 2}$.
\end{remark}

\subsection{$g=1$. The Somos-5 recurrence}

It should be noted that the relation (\ref{896432000096}) is a well-known and fairly well-studied Somos-5 equation. It represents one of the bilinear (quadratic) equations sharing the Laurent property \cite{Gale}. 
The usual notation of this equation in the literature is
\begin{equation}
t_nt_{n+5}=\alpha t_{n+1}t_{n+4}+\beta t_{n+2}t_{n+3},
\label{887854327}
\end{equation}
where  $\alpha$ and $\beta$ are supposed to be some parameters. 
\begin{lemma} \cite{Hone2} 
The rational function
\begin{equation}
J=\frac{t_nt_{n+3}}{t_{n+1}t_{n+2}}+\frac{t_{n+1}t_{n+4}}{t_{n+2}t_{n+3}}+\alpha\left(\frac{t_{n+1}t_{n+2}}{t_nt_{n+3}}+\frac{t_{n+2}t_{n+3}}{t_{n+1}t_{n+4}}\right)+\beta \frac{t_{n+2}^2}{t_nt_{n+4}} 
\label{0945321}
\end{equation}
is an invariant for the Somos-5 equation (\ref{887854327}).
\end{lemma}
\begin{remark}
It is a well-known fact that the invariant $J$ given by (\ref{0945321}) can be represented as $J=\left(\beta_n\beta_{n+1}-\beta\right)/\alpha$, where  
\[
\beta_n=\frac{t_nt_{n+4}+\alpha t_{n+2}^2}{t_{n+1}t_{n+3}}
\] 
turns out to be a 2-invariant for the Somos-5 recurrence, that is, it satisfies $\beta_{n+2}=\beta_n$ in virtue of (\ref{887854327}). In other words, the Somos-5 recurrence (\ref{887854327}) is equivalent to the equation
\[
t_nt_{n+4}=\beta_n t_{n+1}t_{n+3}-\alpha t_{n+2}^2,\;\; \beta_{n+2}=\beta_n,
\]
which is the Somos-4 recurrence with one 2-periodic coefficient.
\end{remark}

Let us rewrite the first relation in (\ref{43296665325}) from  Example \ref{9084532} as an equation
\[
u_{n+1}\left(u_n+u_{n+1}+u_{n+2}-H_1\right)+\alpha=0.
\]
Nothing prevents us to substitute (\ref{896432097}) into this relation. As a result, we obtain a homogeneous equation
\begin{equation}
t_{n+2}t_{n+3}\left(t_{n+1}t_{n+4} H_1-\alpha t_{n+2}t_{n+3}\right)=B_n,
\label{7675322}
\end{equation}
where the discrete polynomial $B_n$ is given by (\ref{87766}). Note that the appearance of the polynomial $B_n$ in (\ref{7675322}), generally speaking,  is unexpected since it arises in  Lemma  \ref{6549098}  corresponding to the case $g=2$ and not $g=1$. 

The following statement is proved by direct verification.
\begin{lemma}
\label{9043765}
The recurrence (\ref{7675322}) is equivalent to Somos-5 equation (\ref{887854327}).   The identification is determined by means of $H_1=J$, where $J$ is given by (\ref{0945321}). 
\end{lemma}
Next we will present a proof of the following statement.
\begin{proposition}
\label{98543209}
The recurrence (\ref{7675322}) has the Laurent property. More precisely, for
all $n\in\mathbb{Z}$, $t_n\in\mathbb{Z}[\alpha, \beta,  t_0, t_1^{\pm 1}, t_2^{\pm 1}, t_3^{\pm 1}, t_4]$.
\end{proposition}

\subsection{Proof of the Proposition \ref{98543209}}

For further discussion it is convenient to introduce the short notation $\mathcal{R}=\mathbb{Z}[\alpha, H_1,  t_0, t_1^{\pm 1}, t_2^{\pm 1}, t_3^{\pm 1}, t_4]$. To be fair, we note that the idea of the proof of Proposition \ref{98543209} presented below is not original and can be found in  \cite{Hone1}.  To prove this proposition, one uses the birational transformation $\left(u_0, u_1,\ldots\right)\leftrightarrow \left(h_1, h_2,\ldots\right)$ defined by the relation
\[
F_0=1+\frac{u_0 x}{1+\displaystyle{\frac{u_1 x}{1+\displaystyle{\frac{u_2 x}{1+\cdots}}}}}=1-\sum_{j\geq 1} (-1)^j h_j.
\]
Provided (\ref{896432097}), we have, for example, 
\begin{equation}
h_1=\frac{t_0t_3}{t_1t_2}\;\; \mbox{and}\;\; h_2=\frac{t_0t_4}{t_2^2}.
\label{900098}
\end{equation}
Let us define  the sequence $\left(\Delta_k\right)_{k\geq 1}$ by $\Delta_{2k-1}=\det \left(h_{n+m-1}\right)_{n, m=1,\ldots, k}$ and $\Delta_{2k}=\det \left(h_{n+m}\right)_{n, m=1,\ldots, k},\; \forall k\geq 1$. 
Additionally, for convenience, let $\Delta_{-2}=\Delta_{-1}=\Delta_0=1$.

The following technical statement will be helpful.
\begin{lemma}
\label{66654320987}
\cite{Hone1}
The relations
\begin{equation}
t_{2k}=t_0\left(\frac{t_2}{t_0}\right)^k\Delta_{2k-2}\;\; \mbox{and}\;\; t_{2k+1}=\frac{t_0t_1}{t_2}\left(\frac{t_2}{t_0}\right)^{k+1}\Delta_{2k-1},\;\; \forall k\geq 0
\label{986453}
\end{equation}
are identities.
\end{lemma}

Now, to prove  Proposition \ref{98543209}, we  note that, in virtue of  Theorem \ref{897643098}, in the case $g=1$,  the sequence $(h_n)_{n\geq 1}$ satisfies 
\begin{equation}
h_n=\left(q_1-p_1\right)h_{n-1}+\sum_{j=1}^{n-1}h_jh_{n-j}-\frac{q_1}{2} \sum_{j=1}^{n-2}h_jh_{n-j-1},\;\; n\geq 3,
\label{6654230986}
\end{equation}
where it is assumed that the pair of initial values $\left(h_1, h_2\right)$ is given by (\ref{900098}). Notice that, by (\ref{900098}), $h_1, h_2\in  t_0 \mathcal{R}$. As for the coefficients, entering  in the recurrence (\ref{6654230986}), then according to  (\ref{90632453}), we have the following expressions: 
\[
q_1-p_1=-2u_0+H_1=-2h_1+H_1=-2\frac{t_0t_3}{t_1t_2}+H_1
\]
and
\[
-\frac{q_1}{2}=u_{-1}+u_0-H_1=-\frac{\alpha+h_2}{h_1}=-t_1\frac{t_0t_4 + \alpha t_2^2}{t_0t_2t_3}.
\]
Thus, we have $q_1-p_1\in  \mathcal{R}$ and $q_1/2\in  \mathcal{R}/t_0$.  By  induction, we obtain the fact that, in virtue of (\ref{900098}) and (\ref{6654230986}),  $h_n\in t_0 \mathcal{R},\; \forall n\geq 1$. Now, it is time to use Lemma \ref{66654320987}. According to what we already have and this lemma, one gets $\Delta_{2k-1 }, \Delta_{2k}\in  t_0^k\mathcal{R},\; \forall k\geq 1$.
Making of use  (\ref{986453}), we obtain that $t_n\in \mathcal{R},\; \forall n\geq 0$. To prove this proposition completely it remains to show that $t_{-n}\in \mathcal{R},\; \forall n\geq 1$. To do this, one may use reversing symmetry: $t_{-n}\cong t_{n+4},\; \forall n\geq 1$, where `$\cong$' means equality provided that reversing $\left(t_0, t_1, t_2, t_3, t_4\right)\rightarrow \left(t_4, t_3, t_2, t_1, t_0\right)$ is performed. Thus,  Proposition \ref{98543209} can be considered proven.

In virtue of  Lemma \ref{9043765} and  Proposition \ref{98543209}, we get the well-known statement \cite{Fomin}.
\begin{proposition}
\label{09432900008}
The Somos-5 recurrence (\ref{887854327}) has the Laurent property. More precisely, for
all $n\in\mathbb{Z}$, $t_n\in\mathbb{Z}[\alpha, \beta,  t_0^{\pm 1}, t_1^{\pm 1}, t_2^{\pm 1}, t_3^{\pm 1}, t_4^{\pm 1}]$.
\end{proposition}
\begin{remark}
It may seem strange that we are re-proving a well-known fact. We mean Proposition \ref{09432900008}. This is all the more so because there is a more meaningful proof for it \cite{Fomin}. The point is that in this example we show a scheme for proving the Laurent property that works for recurrences (\ref{87621112342}) and (\ref{00432000985436}) and some others for $g\geq 4$ for which there is no such meaningful proof.
\end{remark}
\begin{remark}
Let us recall that the Somos-5 recurrence (\ref{887854327}), in fact, represents  a particular case a more general equation also possessing the Laurent property is known to be the three-term Gale-Robinson equation \cite{Gale}
\begin{equation}
t_nt_{n+N}=\alpha t_{n+p}t_{n+N-p}+\beta t_{n+q}t_{n+N-q}
\label{8800987854327}
\end{equation}
with arbitrary $N\geq 4$ and $1\leq p<q\leq \left\lfloor N/2 \right\rfloor$. In   (\ref{8800987854327}), $\alpha$ and $\beta$ are supposed to be  arbitrary parameters. In the cases $N=4$ and $N=5$, as is easy to understand, we have one equation each, which, in the literature, are called Somos-4 and Somos-5, respectively.
Recall that the Gale-Robinson equation (\ref{8800987854327}) arises in the theory of cluster algebras by Fomin and Zelevinsky \cite{Fordy}, \cite{Fordy1}. The Laurent property for the Gale-Robinson recurrence is that for all $n\in\mathbb{Z}$, $t_n\in\mathbb{Z}[\alpha, \beta,  t_0^{\pm 1},\ldots,  t_{N-1}^{\pm 1}]$. This statement is formulated as Theorem 1.8 in  \cite{Fomin}.
\end{remark}

\subsection{$g=2$. Recurrence (\ref{87621112342})} 

Let us apply the following notations: $\alpha=\mathcal{G}_{2, 0}$, $\beta=\mathcal{G}_{2, 2}$ and $\gamma=-\mathcal{G}_{2, 1}$. Then  relation (\ref{87621112342}) looks as
\begin{equation}
t_{n+3}t_{n+4}\left(t_nt_{n+7}-\alpha t_{n+1}t_{n+6}-\beta t_{n+2}t_{n+5}\right)=\gamma B_{n+1}. 
\label{8762119084312342}
\end{equation}
It is obvious that if the condition $\gamma=0$ met, equation (\ref{8762119084312342}) is reduced to the  Gale-Robinson equation  (\ref{8800987854327}), with $\left(N, p, q\right)=(7, 1, 2)$. 

The following lemmas are proved by direct verification.
\begin{lemma} 
\label{809432}
Recurrence (\ref{8762119084312342}) has  invariants
\begin{eqnarray}
J_1&=&\frac{t_nt_{n+3}}{t_{n+1}t_{n+2}}+\frac{t_{n+1}t_{n+4}}{t_{n+2}t_{n+3}}+\frac{t_{n+2}t_{n+5}}{t_{n+3}t_{n+4}}+\frac{t_{n+3}t_{n+6}}{t_{n+4}t_{n+5}} \nonumber\\
&&+\alpha \left(\frac{t_{n+1}t_{n+4}}{t_{n}t_{n+5}}+\frac{t_{n+2}t_{n+5}}{t_{n+1}t_{n+6}}\right)+\beta \frac{t_{n+2}t_{n+4}}{t_{n}t_{n+6}} \nonumber\\
&& +\gamma\left(\frac{t_{n+2}t_{n+3}}{t_{n}t_{n+5}}+\frac{t_{n+3}t_{n+4}}{t_{n+1}t_{n+6}}+\frac{t_{n+1}t_{n+4}^2}{t_{n}t_{n+3}t_{n+6}}+\frac{t_{n+2}^2t_{n+5}}{t_{n}t_{n+3}t_{n+6}}\right) 
\label{78635642}
\end{eqnarray}
and
\begin{eqnarray}
J_2&=&\frac{t_{n}t_{n+5}}{t_{n+1}t_{n+4}}+\frac{t_{n+1}t_{n+6}}{t_{n+2}t_{n+5}}+\frac{t_{n}t_{n+3}^2t_{n+6}}{t_{n+1}t_{n+2}t_{n+4}t_{n+5}} \nonumber\\
&&+\alpha\left(\frac{t_{n+1}t_{n+2}}{t_{n}t_{n+3}}+\frac{t_{n+2}t_{n+3}}{t_{n+1}t_{n+4}}+\frac{t_{n+3}t_{n+4}}{t_{n+2}t_{n+5}}+\frac{t_{n+4}t_{n+5}}{t_{n+3}t_{n+6}}+\frac{t_{n+1}^2t_{n+4}^2}{t_{n}t_{n+2}t_{n+3}t_{n+5}} \right. \nonumber\\
&&\left. +\frac{t_{n+2}^2t_{n+5}^2}{t_{n+1}t_{n+3}t_{n+4}t_{n+6}}\right)+\beta \left(\frac{t_{n+1}t_{n+4}^2}{t_{n}t_{n+3}t_{n+6}}+\frac{t_{n+2}^2t_{n+5}}{t_{n}t_{n+3}t_{n+6}}\right)   \nonumber\\
&&+\gamma \left(\frac{t_{n+2}^2}{t_{n}t_{n+4}}+\frac{t_{n+3}^2}{t_{n+1}t_{n+5}}+\frac{t_{n+4}^2}{t_{n+2}t_{n+6}}+\frac{t_{n+1}t_{n+4}}{t_{n}t_{n+5}}+\frac{t_{n+2}t_{n+5}}{t_{n+1}t_{n+6}} \right. \nonumber\\
&&\left. +\frac{t_{n+2}^3t_{n+5}^2}{t_{n}t_{n+3}^2t_{n+4}t_{n+6}}+\frac{t_{n+1}^2t_{n+4}^3}{t_{n}t_{n+2}t_{n+3}^2t_{n+6}} +2 \frac{t_{n+1}t_{n+2}t_{n+4}t_{n+5}}{t_{n}t_{n+3}^2t_{n+6}}\right).
\label{76098342}
\end{eqnarray}
\end{lemma} 
\begin{lemma}
\label{675432}
Recurrence (\ref{8762119084312342}) is equivalent to the following one:
\begin{equation}
t_{n+2}t_{n+3}t_{n+4}t_{n+5}\left(B_{n+1}H_1-t_{n+2}t_{n+3}t_{n+4}t_{n+5} H_2+\alpha t_{n+3}^2t_{n+4}^2\right)=D_n,
\label{887888657}
\end{equation}
where $B_n$ is given by (\ref{87766}) and
\begin{eqnarray*}
D_n&=&t_nt_{n+3}^2t_{n+4}^3t_{n+5}^2+t_{n+2}^2t_{n+3}^3t_{n+4}^2t_{n+7}+t_{n+1}^2t_{n+4}^4t_{n+5}^2+t_{n+2}^2t_{n+3}^4t_{n+6}^2 \\
&&+2t_{n+1}t_{n+2}^2t_{n+4}^2t_{n+5}^3+2t_{n+2}^3t_{n+3}^2t_{n+5}^2t_{n+6} \\
&&+t_{n+1}t_{n+2}t_{n+3}^2t_{n+4}^2t_{n+5}t_{n+6}+t_{n+2}^4t_{n+5}^4.
\end{eqnarray*}
The identification is determined by means of $H_1=J_1$ and $H_2=J_2$, where  $J_1$ and $J_2$ are given by (\ref{78635642}) and (\ref{76098342}), respectively. 
\end{lemma}
One sees that  recurrence (\ref{8762119084312342}) looks less cumbersome compared to its equivalent (\ref{887888657}). 
The recurrence (\ref{887888657}) was written out in  \cite{Hone1} and there one can  find regarding it the following statement.
\begin{proposition}
\label{897432}
Recurrence (\ref{887888657}) has the Laurent property. More precisely, for
all $n\in\mathbb{Z}$, $t_n\in\mathbb{Z}[\alpha, H_1,  H_2,  t_0, t_1, t_2^{\pm 1}, t_3^{\pm 1}, t_4^{\pm 1}, t_5, t_6]$.
\end{proposition}
As a consequence of  Lemma \ref{675432} and  Proposition \ref{897432},  we get the following statement.
\begin{proposition}
Recurrence (\ref{8762119084312342}) has the Laurent property. More precisely, for
all $n\in\mathbb{Z}$, $t_n\in\mathbb{Z}[\alpha, \beta, \gamma, t_0^{\pm 1}, t_1^{\pm 1}, t_2^{\pm 1}, t_3^{\pm 1}, t_4^{\pm 1}, t_5^{\pm 1}, t_6^{\pm 1}]$.
\end{proposition}
\begin{remark}
We borrowed the idea of proving   Proposition \ref{897432}, in  \cite{Hone1}, to prove  Proposition \ref{98543209} above.
\end{remark}
The following remark is in order. Obviously, recurrence (\ref{8762119084312342}) is not bilinear although it is homogeneous. However, Proposition 2.1 in \cite{Hone1} tells us that any sequence satisfying (\ref{887888657}) also satisfies some bilinear recurrence of order 9 with suitable coefficients. Taking into account this proposition and Lemma \ref{675432}, we are able to formulate the following statement.
\begin{proposition}
\label{76540}
Any sequence satisfying (\ref{8762119084312342}) also satisfies the Somos-9 recurrence 
\[
t_nt_{n+9}=\alpha_1 t_{n+1}t_{n+8}+\alpha_2 t_{n+2}t_{n+7}+\alpha_3 t_{n+3}t_{n+6}+\alpha_4 t_{n+4}t_{n+5}
\]
with coefficients
\[
\alpha_1=\gamma-\frac{\alpha\beta}{\gamma},\;\; \alpha_2 = \alpha\left(\frac{\alpha \beta}{\gamma}-2\gamma\right),\; \alpha_3 = \alpha\left(\frac{ \beta^2}{\gamma}+\beta J_1 + \gamma J_2  +\alpha^2\right),\; \alpha_4=\frac{\gamma}{\alpha}\alpha_3,
\]
where $J_1$ and $J_2$ are given by (\ref{78635642}) and (\ref{76098342}), respectively. 
\end{proposition}
\begin{remark}
Obviously, the converse, generally speaking, is not true. Note that a sequence satisfying (\ref{8762119084312342}) may also  satisfies  bilinear relation 
\[
t_{n}t_{n+N}=\sum_{j=1}^{\left\lfloor N/2 \right\rfloor} \alpha_j t_{n+j} t_{n+N-j}
\]
of higher order than 9, but the value 9 is apparently the lowest one.
\end{remark}

\subsection{$g=3$. Recurrence (\ref{00432000985436})}

Similar proposition concerning the Laurent property for  recurrence (\ref{00432000985436}) which we  rewrite in the form 
\begin{equation}
t_{n+3}t_{n+4}t_{n+5}t_{n+6}\left(t_nt_{n+9}-\alpha t_{n+1}t_{n+8}-\beta t_{n+2}t_{n+7}\right)=\gamma_1 B_{n+1, 1}+\gamma_2 B_{n+1, 2},
\label{8875409327}
\end{equation}
where $B_{n, 1}$ and $B_{n, 2}$ are given by (\ref{896432000985436}) and (\ref{89677766}), respectively. Here we have designated $\alpha=-\mathcal{G}_{3, 0}$, $\beta=\mathcal{G}_{3, 3}$, $\gamma_1=\mathcal{G}_{3, 1}$ and $\gamma_2=-\mathcal{G}_{3, 2}$. Again we see that  recurrence (\ref{8875409327}) is a generalization of the Gale-Robinson equation (\ref{8800987854327}) with $\left(N, p, q\right)=(9, 1, 2)$.  In this case, we are also able to write an equivalent equation that contains the parameters $(\alpha, H_1, H_2, H_3)$ and prove the  Laurent property using the established scheme, but the corresponding formulas look too cumbersome and  we would not like to clutter up the space with these formulas. Let us just formulate the final result.
\begin{proposition}
The recurrence (\ref{8875409327}) has the Laurent property. More precisely, for
all $n\in\mathbb{Z}$, $t_n\in\mathbb{Z}[\alpha, \beta, \gamma, \delta, t_0^{\pm 1}, t_1^{\pm 1}, t_2^{\pm 1}, t_3^{\pm 1}, t_4^{\pm 1}, t_5^{\pm 1}, t_6^{\pm 1}, t_7^{\pm 1}, t_8^{\pm 1}]$.
\end{proposition}
\begin{remark}
There must exist some smallest order $N$ such that a kind of Proposition \ref{76540} is true concerning the recurrence (\ref{8875409327}). Based on actual calculations, we assume that in this case $N=17$.
\end{remark}

\subsection{Examples of integer sequences}

The Laurent property guarantees the integerness of the members of sequence satisfying  (\ref{887854327}), (\ref{8762119084312342}) or (\ref{8875409327}) with a suitable choice of parameters included in these relations and initial conditions. Let us show some integer sequences satisfying these homogeneous equations.
\begin{example}
Let $\alpha=\beta=1$ and $t_n=1$ for $0\leq n\leq 4$. In this case, a sequence defined by the Somos-5 recurrence (\ref{887854327}), starting from $t_5$,  continues like this:
\[
2,\; 3,\; 5,\; 11,\; 37,\; 83,\; 274,\;  1217,\; 6161,\; 22833,\; 165713,\; 1249441,\ldots 
\]
This sequence is known as \textrm{A006721} in OEIS. It is important to note that it has a  variety of different applications (see, for exaple \cite{Buchholz}, \cite{Eager}, \cite{Fordy}).
\end{example}
\begin{example}
Let $\alpha=\beta=\gamma=1$ and $t_n=1$ for $0\leq n\leq 6$. The  sequence defined by the recurrence (\ref{8762119084312342}), starting from $t_7$,  continues like this:
\[
5,\; 13,\; 61,\; 185,\; 533,\; 3973,\; 56161,\; 581197,\; 10521521,\; 75458197,\ldots
\]
The first six prime numbers in this sequence are 
\[
t_7=5,\; t_8=13,\; t_9=61,\; t_{14}=581197,\; t_{16}=75458197,
\]
\[
t_{28}=471568036896759616693823298697.
\]
\end{example}
\begin{example}
Let $\alpha=\beta=\gamma_1=\gamma_2=1$ and $t_n=1$ for $0\leq n\leq 8$. The  sequence defined by the recurrence (\ref{8875409327}), starting from $t_9$,  continues like this:
\[
13,\; 73,\; 937,\; 5437,\; 63853,\; 458353,\; 2506657,\; 64527157,\ldots
\]
The first seven prime numbers in this sequence are
\[
t_9=13,\; t_{10}=73,\; t_{11}=937,\; t_{12}=5437,\; t_{13}=63853,\; t_{15}=2506657,
\]
\[
t_{39}=1319274733570152642028977015416314389776808557991387771172617237650793241.
\]
\end{example}

\section{Final remarks}

Let us spend  some lines slightly deviating from the main content of the article to make some remarks. 

\subsection{On discrete polynomials}

In this work we mainly use discrete polynomials $S^k_s(n)$, but their companions $T^k_s(n)$ can also be used in some situations.
It is known that polynomials $S^k_s(n)$ and $T^k_s(n)$  are connected to each other via the identities  \cite{Svinin1}
\[
S^k_s(n)+\sum_{j=1}^{k-1} (-1)^j S^{k-j}_s(n+j)T^j_s(n)+(-1)^kT^k_s(n)=0 
\]
and
\[
T^k_s(n)+\sum_{j=1}^{k-1} (-1)^j T^{k-j}_s(n+j)S^j_s(n)+(-1)^kS^k_s(n)=0,\;\; k, s\geq 1. 
\]
Perhaps discrete polynomials $S^k_s(n)$ and $T^k_s(n)$  should be understood as a kind of analogue of the complete and elementary symmetric polynomials $h_k$ and $e_k$, respectively, that satisfy an identity
\[
e_k+\sum_{j=1}^{k-1} (-1)^j e_{k-j}h_j+(-1)^kh_k=0,\;\; \forall  k\geq 1.
\]
It is known that symmetric polynomials satisfy many identities and it might be interesting to search for analogues of these identities for discrete polynomials $S^k_s(n)$ and $T^k_s(n)$. Well, as an example, in  \cite{Svinin3}  the identities 
\[
S^k_s(n)=(-1)^k\sum_{K} (-1)^p T^{k_1}_s(n) T^{k_2}_s(n+k_1)\cdots T^{k_p}_s(n+k_1+\cdots+k_{p-1})
\]
and
\[
T^k_s(n)=(-1)^k\sum_{K} (-1)^p S^{k_1}_s(n) S^{k_2}_s(n+k_1)\cdots S^{k_p}_s(n+k_1+\cdots+k_{p-1}),
\]
for all $s\geq 1,$ has been presented.  The summation in these two formulas is performed over all compositions $K= (k_1,\ldots, k_p)$ of number $k$ with $k_j\geq 1$. For instance, in the case $k=2$, we have the identities
\[
S^2_s(n)=-T^2_s(n)+T^1_s(n)T^1_s(n+1)\;\; \mbox{and}\;\; T^2_s(n)=-S^2_s(n)+S^1_s(n)S^1_s(n+1).
\]

We also note that in  \cite{Svinin1} broader classes of discrete polynomials were actually defined and used. For instance, given any natural number $q$, let
\[
S^k_s(n)=\sum_{0\leq \lambda_1\leq \cdots\leq  \lambda_k\leq s-1} \prod_{j=1}^k u_{n+\lambda_j+(k-j)q}. 
\]
These polynomials can  be defined again using a pseudo-difference formal operator a kind of  (\ref{67888}). Namely, let us write
\[
\mathcal{T}=\Lambda^{-1} + u_{n-1} x \Lambda^{-1-q}.
\]
The formula for determining $S^k_s(n)$, in this more general case, has changed slightly compared to (\ref{678760}) and looks like 
\[
\mathcal{T}^{-s}=\Lambda^s+\sum_{j\geq 1} (-1)^j x^j S^j_s(n-(j-1)q) \Lambda^{s-jq}.
\]
One of the results of the paper \cite{Svinin1} is that it has been shown  that the integrable hierarchy of evolution differential-difference equation
\begin{equation}
\frac{d u_n}{d t}=u_n\left(\sum_{j=1}^q u_{n+j}- \sum_{j=1}^q u_{n-j}\right)
\label{67587776542}
\end{equation}
can be written in explicit form as
\begin{equation}
\frac{d u_n}{d t_s}=u_n\left(S^s_{sq}(n-(s-1)q+1)-S^s_{sq}(n-sq)\right).
\label{6777772}
\end{equation}

Generally speaking, the goal of the work \cite{Svinin1} was that equation (\ref{67587776542}) and its hierarchy of generalized symmetries as a whole can be obtained as a result of a suitable reduction to a bi-infinite sequence of Kadomtsev-Petviashvili  hierarchies presented in a convenient form. Also within the framework of this approach, it turned out to be possible to obtain an infinite classes of constraints that are compatible with the  hierarchy of  symmetries (\ref{6777772}). Later in \cite{Svinin3} these ideas has been extended to many systems of evolution equations sharing the property of having a representation in the form of a Lax pair. 

As a result, explicit forms of their integrable hierarchies were obtained in terms of suitable discrete polynomials. One way or another, for each specific case, it is necessary to find a suitable pseudo-difference operator like (\ref{67888}). We will not describe everything in detail further and refer the reader to  \cite{Svinin3} for details. 

\subsection{On homogeneous equations}

Let us also say a few words about the homogeneous equations that arise in our work. A known  proof of the Laurent property for the Gale-Robinson recurrence (\ref{8800987854327}), including the Somos-5 equation (\ref{887854327}), is closely related to its application in the theory of cluster algebras \cite{Fomin}. The recurrences (\ref{8762119084312342}) and (\ref{8875409327}) are apparently not related to the theory of cluster algebras. It would be interesting to find applications for them and prove the Laurent property in a more meaningful way.

Inspired by the examples given above, it can be assumed that for any $g\geq 2$ the homogeneous recurrence on the sequence $\left(t_n\right)$ might look like
\begin{equation}
\prod_{j=3}^{2g} t_{n+j}\cdot \left(t_nt_{n+2g+3}-\alpha t_{n+1}t_{n+2g+2}-\beta t_{n+2}t_{n+2g+1}\right)=\sum_{j=1}^{g-1}\gamma_j B_{n+1, j},
\label{8878886500097}
\end{equation}
where $\left(B_{n, 1},\ldots, B_{n, g-1}\right)$ is supposed to be  a finite set of homogeneous polynomials in variables $\left(t_n,\ldots, t_{n+2g+1}\right)$ of degree $2g$. These polynomials must be symmetric in such a way  the hypothetical recurrence (\ref{8878886500097}) to be invariant under the reversing $\left(t_n, t_{n+1},\ldots, t_{n+2g+3}\right)\rightarrow \left(t_{n+2g+3}, t_{n+2g+2},\ldots, t_n\right)$. Thus we assume that these recurrences must represent a kind of generalization of the Gale-Robinson equation (\ref{8800987854327}) of type $\left(2g+3, 1, 2\right)$.

\section*{Acknowledgments}

The results were obtained within the framework of the state assignment of the Ministry of Education and Science of the Russian Federation on the project No. 121041300058-1.


\begin{thebibliography}{99}



\bibitem{Buchholz}
Buchholz R~H,   Rathbun R~L 
1997
An infinite set of Heron triangles with two rational medians, 
{\em Amer. Math. Monthly} \textbf{104}  107--115.

\bibitem{Gubbiotti}
Gubbiotti G, Joshi N, Tran D T and Viallet C-M 
2020 
Bi-rational maps in four dimensions with two invariants 
{\em J. Phys. A: Math. Theor.} \textbf{53} 115201.

\bibitem{Eager}
Eager R, Franco S 
2012
Colored BPS pyramid partition functions, quivers and cluster transformations
{\em JHEP} \textbf{09} 038.

\bibitem{Fomin}
Fomin S, Zelevinsky A 
2002
The Laurent phenomenon. 
{\em Adv. Appl. Math.}  \textbf{28}  119--144.

\bibitem{Fordy}
Fordy A~P, March B~R 
2011
Cluster mutation-periodic quivers and associated Laurent sequences
{\em J. Algebr. Comb.} \textbf{34} 19--66.

\bibitem{Fordy1}
Fordy A~P, Hone A~N~W 
2014
Discrete integrable systems and Poisson algebras from cluster maps
{\em Commun. Math. Phys.} \textbf{325} 527--584.


\bibitem{Gale}
Gale D 
1991
The strange and surprising saga of the Somos sequences 
{\em Math. Intell.} \textbf{13}  40--42.

\bibitem{Hone2}
Hone A~N~W 
2007 
Sigma function solution of the initial value problem for Somos 5 sequences
{\em Trans. Amer. Math. Soc.} \textbf{359} 5019--5034.

\bibitem{Hone3}
Hone A~N~W  
2022
Heron triangles with two rational medians and Somos-5 sequences. 
{\em European J.  Math.} \textbf{8} 1424--1486.

\bibitem{Hone1}
Hone A~N~W, Roberts J~A~G and  Vanhaecke P 
2024
A family of integrable maps associated with
the Volterra lattice
{\em Nonlinearity} \textbf{37}  095028, 65 pp.

\bibitem{Kac}
Kac M,  van Moerbeke P 
1975 
On an explicitly soluble system of nonlinear differential equations
related to certain Toda lattices 
{\em Adv. Math.} \textbf{3} 160--169.

\bibitem{Makdonald}
Makdonald I~G 1995 
Symmetric Functions and Hall Polynomials 
2nd edn (Oxford: Oxford University Press)

\bibitem{Manakov}
Manakov S~V 
1974
Complete integrability and stochastization of discrete dynamical systems 
{\em Zh. Exp. Teor. Fiz.} 
\textbf{67.2}  543--555.

\bibitem{Mumford}
Mumford D 1984 Tata Lectures on Theta II (Birkh\"auser)

\bibitem{Stieltjes}
Stieltjes T~J. 
1894
Recherches sur les fractions continues
{\em Annales de la Facult\'e des sciences de Toulouse: Math\'ematiques} 
\textbf{8} No. 4. 

\bibitem{Svinin1}
Svinin A~K 
2009
On some class of reductions for the Itoh-Narita-Bogoyavlenskii lattice
{\em J.  Phys. A: Math. Theor.}  
\textbf{42} 454021, 15pp.

\bibitem{Svinin2}
Svinin A~K 
2011
On some class of homogeneous polynomials and explicit form of integrable hierarchies of differential-difference equations
{\em J.  Phys. A: Math. Theor.} 
\textbf{44}   165206, 16 pp.

\bibitem{Svinin3}
Svinin A~K 
2014
On some classes of discrete polynomials and ordinary difference equations
{\em J.  Phys. A: Math. Theor.} \textbf{47} (2014)  155201, 27 pp.

\bibitem{Svinin4}
Svinin A K 
2016
On integrals for some class of ordinary difference equations admitting a Lax representation
{\em J.  Phys. A: Math. Theor.} 
\textbf{49}   095201, 34 pp.

\bibitem{Svinin5}
Svinin A K 
2021
On solutions for some class of integrable difference equations
{\em J. Differ. Equ.  Appl.} \textbf{27}  1734--1750.



\bibitem{Vanhaecke}
Vanhaecke P. 
2001
Integrable systems in the realm of algebraic geometry
{\em Springer Science and Business Media}.


\end{thebibliography}
\end{document}